\documentclass[11pt,a4paper]{article}
\usepackage[UKenglish]{babel}
\usepackage[utf8]{inputenc}
\usepackage{float}   
\usepackage{url}     
\usepackage{geometry}
\usepackage{enumitem}  
\usepackage{amsmath}
\usepackage{amssymb}
\usepackage{graphicx}
\usepackage[T1]{fontenc}
\usepackage{color}
\usepackage[flushleft]{threeparttable}
\usepackage{tabularx} 
\usepackage[justification=centering]{caption}
\usepackage{subcaption}
\usepackage{color,soul}
\usepackage{datetime}
\usepackage{pdflscape}
\usepackage{afterpage}
\usepackage{adjustbox}
\usepackage{booktabs}
\usepackage{capt-of}
\usepackage{setspace}
\usepackage{rotating}

\allowdisplaybreaks
\usepackage[
authordate,
bibencoding=auto,
strict,
backend=biber,
natbib,
maxbibnames=99,
maxcitenames=2,
doi=false,
isbn=false,
url=true, 
eprint=false
]{biblatex-chicago}
\AtEveryBibitem{%
	\ifentrytype{article}{%
		\clearfield{url}%
		\clearfield{urldate}%
		\clearfield{urlyear}%
		\clearfield{urlmonth}%
		\clearfield{urlday}%
	}{}%
}
\usepackage{amsthm}
\newtheoremstyle{wooldridge}
{}{}                                         
{\normalfont}                                
{\parindent}                                 
{\bfseries}                                  
{:}                                          
{ }                                          
{\thmname{#1}\thmnumber{ #2}\thmnote{ (#3)}} 
\theoremstyle{wooldridge}
\newtheorem{innerassumption}{Assumption}
\newcommand{\inlineqed}{\unskip\nobreak\ \qedsymbol}
\newenvironment{assumption}[1][]
{\pushQED{\inlineqed}%
	\begin{innerassumption}[#1]}
	{\popQED\end{innerassumption}}
\makeatletter
\newcommand{\assumptionlabel}[2]{%
	\def\@currentlabel{#2}%
	\label{#1}%
}
\makeatother

\theoremstyle{plain}   
\newtheorem{proposition}{Proposition}

\newtheorem{remark}{Remark}

\theoremstyle{definition}
\newtheorem{innerdefinition}{Definition}
\newenvironment{definition}[1][]
{\pushQED{\inlineqed}%
	\begin{innerdefinition}[#1]}
	{\popQED\end{innerdefinition}}
\theoremstyle{plain}

\newdateformat{monthyeardate}{\monthname[\THEMONTH] \THEYEAR}
\newcommand{\E}{\mathbb{E}}

\newcommand{\plim}{\text{plim}}

\usepackage{verbatim}

\title{First as Tragedy? Second as What? \\ Estimating Dynamic Effects of Recurrent Events} 
\author{Agnes Norris Keiller\thanks{London School of Economics}}
\date{\monthyeardate\today }

\newcommand{\mcnsim}{1000}
\newcommand{\mcnbreps}{500}
\newcommand{\mcmaxh}{8}

\newcommand{\mcsieBiasThreeBaseTwo}{-0.1362}

\newcommand{\mcsieCovThreeBaseTwo}{0.372}

\newcommand{\mcmeAdvantageThree}{three}
\newcommand{\mcmeBiasFourBaseTwo}{+0.5028}

\newcommand{\mcmeBiasTwoTwo}{-0.3074}
\newcommand{\mcsieBiasTwoTwo}{-0.0015}
\newcommand{\mcmeBiasThreeBaseTwo}{-0.0454}

\newcommand{\mcrefusalRateThreeSat}{4.1 to 5.3\%}

\begin{document}

\vspace{-3cm}
\maketitle
	
\begin{abstract}
    I study treatment effect estimation when treatment events have persistent effects and can be experienced more than once. Natural disasters, job loss and health shocks are examples of such treatments. I show that the effect of a total treatment trajectory can be recovered under assumptions similar to those commonly invoked in single-event settings using suitably flexible TWFE models. Decomposing the total trajectory effect into portions attributable to distinct event occurrences, however, requires further assumptions. I propose an assumption similar to conditional parallel trends, imposing it on the growth of event-specific effects rather than on untreated outcomes. Combined with a linear-in-parameters model of effect growth, this assumption enables a sequential imputation estimator that consistently estimates the dynamic effects of each event occurrence and that can accommodate heterogeneity in effects according to observable event attributes, such as intensity. I demonstrate that several intuitive TWFE models fail to recover interpretable treatment effect parameters in the multi-event setting and illustrate the sequential imputation estimator's favourable performance using Monte Carlo simulations. 
    
    \bigskip
		
    \bigskip
    \newline
    \noindent \textbf{Acknowledgments}: I would like to thank Kirill Borusyak, Robin Burgess, Xavier Jaravel, \'{A}ureo de Paula, Steve Pischke and John Van Reenen for helpful comments and encouragement. Generative AI (Claude, Anthropic) was used during the development of this paper, primarily when implementing Monte Carlo simulations. All errors are my own.
    
\end{abstract}
	
\newpage

\onehalfspacing

\section{Introduction}

There are many examples of `treatments' that units can experience more than once over time. Floods can hit the same location on different dates, people who lose one job may go on to lose another, surviving one negative health shock does not impart immunity to subsequent illness. The effect of such recurrent events is a plausible subject of research interest. One may be particularly interested in how the effect of subsequent events compares to that of initial events. Such comparison would shed light on, for example, whether areas adapt to flooding and whether people become more resilient to job loss or bad health. Such analysis, however, requires one to obtain treatment effect estimates for each event.
	
This paper examines how the dynamic effects of recurrent events can be estimated separately for each event occurrence. I first show how the analysis of \citet{wooldridge_twoway_2025} extends in a straightforward manner to multi-event settings by replacing cohort indicators with indicators that encode the timing of all event occurrences. In this manner, one can consistently estimate the total effect of a treatment trajectory, which reflects the effects of all events experienced up to a particular point in time, using suitably-flexible TWFE specifications under assumptions that are commonly imposed in single-event settings. Decomposing the total trajectory effect into portions attributable to distinct event occurrences, however, requires additional assumptions. For units that have experienced $k$ events at a given point in time, for example, one needs some way of distinguishing the effect of the $k$-th event from the dynamic effects of the prior $k-1$ events. To achieve this decomposition, I propose a `conditional parallel effect-trends' assumption, which imposes that growth in the effects of prior events is conditionally independent of selection into subsequent events -- a condition that can be supported by diagnostic tests similar to `pre-trend' analysis in single-event settings. Under such an assumption, units who have not experienced a $k$-th event at time $t$ can be used to estimate the effects of the prior $k-1$ events and thereby isolate the effect of the $k$-th event for units that have. Together with a linear-in-parameters restriction on how effect growth depends on observables, the assumption enables a sequential imputation estimator (SIE), which provides consistent estimates of the average treatment effect of each event occurrence at each horizon after its impact and can accommodate heterogeneity in treatment effects according to pre-determined characteristics and to observable event attributes, such as event intensity. 
	
To illustrate the advantages of the SIE, I show that several TWFE specifications that have been implemented in recurrent-event settings suffer from a number of issues under different notions of treatment effect heterogeneity. Single-event specifications that include indicators for the first event only can, at best, recover estimates of the total trajectory effect, and fail to do so when the mean trajectory effect at a given first-event horizon varies with the date of the first event. Specifications that pool event indicators across occurrences, such as distributed lag models, require effects to be common across occurrences and conditional mean effects to be invariant to event timing. When these conditions do not hold, their coefficients lack clean interpretation as average effects. Multi-event specifications, which include separate indicators for each event occurrence, accommodate occurrence heterogeneity but are biased when conditional mean effects vary with the timing trajectory, which occurs, for example, when selection into subsequent events depends on the effects of prior events. Monte Carlo simulations demonstrate the empirical importance of these limitations, the favourable performance of the SIE and the ability of the suggested heuristic test to detect failure of the conditional parallel effect-trends assumption on which the sequential imputation estimator depends. 
	
This paper's primary contribution relates to a recent literature on event-study research designs. This `new' event-study literature highlights shortcomings of commonly-implemented estimators and proposes solutions that allow researchers to recover unbiased estimates of various quantities related to treatment effects. The majority of this literature studies settings in which treatment is binary and absorbing, in that units switch into treatment once and remain treated forever. In this setting, \citet{goodman-bacon_difference_2021} and \citet{dechaisemartin_twoway_2020} show treatment effect estimates obtained from standard difference-in-differences regressions reflect weighted averages of heterogeneous treatment effects, with weights that may be negative. \citet{callaway_difference_2021} propose an estimator of group-time average treatment effects that can accommodate time-varying covariates and be used to analyse dynamic treatment profiles. \citet{sun_estimating_2021} also focus on dynamic effects, showing that coefficients on lags and leads of a treatment event indicator are effectively `contaminated' by effects that operate at other lags and leads. \citet{wooldridge_twoway_2025} shows that the cohort-based corrections of \citet{callaway_difference_2021} and \citet{sun_estimating_2021}, as well as the imputation estimator of \citet{borusyak_revisiting_2024}, can equivalently be implemented via OLS estimation of an `extended' TWFE model that controls for treatment cohort dummies and their interaction with time dummies (and possibly additional covariates). 

This paper shows that the analysis of \citet{wooldridge_twoway_2025} applies directly in recurrent-event settings once cohort dummies are redefined to reflect the timing of all event occurrences rather than just the first. The total effect of a given event timing configuration can therefore be consistently estimated under recurrent-event analogues of assumptions conventionally invoked in single-event difference-in-differences analysis. These assumptions do not, however, provide basis for decomposing the total trajectory effect into portions attributable to each event occurrence. To achieve this, I propose a further `conditional parallel effect-trends' assumption, which is similar to the familiar parallel trends assumption but applied to event effects rather than to untreated potential outcomes. Alongside the conventional assumptions, a support requirement, a linearity restriction on effect growth that strengthens this assumption, and standard regularity conditions, I show that the proposed SIE consistently estimates event-specific treatment effects.
	
An exception to the predominant focus on binary and absorbing treatment is \citet{de_chaisemartin_did_2024}, who examine treatment effect estimation when units are subject to different trajectories of continuous treatment. Treatment in their setting can increase or decrease in intensity over the course of a study period, in ways that vary across units. \citeauthor{de_chaisemartin_did_2024} show that conditioning on treatment status in the first period of observation allows one to recover estimates of the impact of the entire trajectory of treatment relative to a `status-quo' counterfactual under which the treatment of each unit stayed at the same level as in the first observed period. This framework is more general than that of the present paper, and one can effectively cast recurrent events into it by using `count of events experienced' as the time-varying treatment intensity. In doing so, the \citeauthor{de_chaisemartin_did_2024} approach is only able to estimate the same total event trajectory effect as one obtains via the trajectory-based extension of the method(s) proposed in \citet{wooldridge_twoway_2025}. Relative to the work of \citeauthor{de_chaisemartin_did_2024}, the estimator proposed in this paper leverages the more structured environment of recurrent events to separately identify the dynamic effects of distinct event occurrences.
	
Prior to the more recent event-study literature discussed above, \citet{sandler_multiple_2014} also explicitly consider event study estimation when treatment events can reoccur. They use Monte Carlo simulations to document the performance of various specifications but do not derive any formal bias results. Relative to their work, this paper contributes formal derivations of the bias in several TWFE specifications and shows the specification they suggest as relatively robust to multiple event occurrences (effectively the multi-event TWFE specification), is vulnerable to bias when treatment effects are heterogeneous across units in ways correlated with event timing.
	
This paper also contributes to empirical literatures that study the effects of recurrent events. This includes work that examines the impacts of natural disasters, of negative health shocks and of job loss. While each of these literatures acknowledges the possibility that units can experience recurrent events over time, none have formally examined how to recover unbiased estimates of distinct events. Researchers have instead adopted a number of methods in the hope that they will recover unbiased estimates of some quantity of interest. These methods include dropping observations that experience more than one event over a short period \citep{erda_cleansing_2026}, assuming dynamic effects are zero after a specific horizon \citep{cengiz_effect_2019} focusing on a single (often the first) event occurrence \citep{deryugina_fiscal_2017,fadlon_family_2019,patel_floods_2024}, or implementing pooled designs that estimate a single set of dynamic effects across event occurrences via distributed lag models \citep{krolikowski_choosing_2018,strobl_economic_2011,hsiang_environment_2014}. I contribute to this work by showing how such previously-adopted strategies can fail to recover unbiased treatment effect estimates when events reoccur and by proposing an estimator that, under certain assumptions, is consistent for the dynamic effects of distinct treatment events in such contexts. 
	
The remainder of this paper proceeds as follows. Section~\ref{sect:environment} formalises the setting in which treatment events can occur more than once. Section~\ref{sect:wooldridge} explains how the results of \citet{wooldridge_twoway_2025} -- and, by implication, all of \citet{callaway_difference_2021,sun_estimating_2021,borusyak_revisiting_2024} -- extend to the recurrent-event setting and highlights why further assumptions are necessary in order to decompose total trajectory effects into event-specific effects. Section~\ref{sect:estimator} proposes the conditional parallel effect-trends assumption and the sequential imputation estimator of average event-specific treatment effects that follows, alongside diagnostic tests that can be used to assess its validity. Section~\ref{sect:twfe} shows how various alternative TWFE specifications fail to recover unbiased estimates of treatment effect parameters in the recurrent-event setting, whereas Section~\ref{sect:montecarlo} presents Monte Carlo simulations that compare their performance to that of the proposed estimator under various DGPs. Section~\ref{sect:implementation} considers several issues regarding empirical implementation and Section~\ref{sect:conclusion} concludes.
	
\section{Environment and Target Estimands}\label{sect:environment}
	
\subsection{The Environment of Study}\label{sect:basic_setup}
Consider a (possibly unbalanced) panel of units $i \in \{1, \cdots, N\}$ observed over a number of periods $t \in \{1, \cdots, T\}$, indexed $(i,t) \in \Omega$. Unit $i$ may experience up to $K$ occurrences of a binary treatment event.\footnote{I focus on binary treatment events for simplicity and include separate discussion, where appropriate, of how results extend to settings where each treatment event occurrence can vary in intensity.} Let $E_i^{(k)}$ denote the period in which unit $i$ experiences its $k$-th event, with $E_i^{(k)} = \infty$ if fewer than $k$ events occur. The treatment trajectory of each unit is summarised by $E_i=\{E_i^{(1)},E_i^{(2)},...,E_i^{(K)}\}$. Simultaneous events are ruled out, so that $E_i^{(k)}<E_i^{(k+1)}$ for every $k$ with $E_i^{(k+1)}<\infty$. Relative time is defined as $\ell_{it}^{(k)} = t - E_i^{(k)}$. 
	
\paragraph{Initial conditions.} The count of experienced events, $k\in\{0, 1, \cdots, K\}$, pertains to events experienced within the period of the observed panel. While it is possible that some units have experienced treatment events prior to the first observed period that may continue to affect observed outcomes, the researcher possesses no information that can be used to identify their effect. Methods to account for this unobserved initial condition problem are discussed in Section \ref{sect:implementation}. Such unobserved events are abstracted from here for notational simplicity by imposing \newline $E_i^{(1)}>1$. This restriction also ensures that all units are observed in an untreated state for at least one period, which is necessary in order to distinguish event effects from permanent unobserved differences between units.
	
\paragraph{Sample partition.} The $(i,t)$ observations can be partitioned into disjoint sets $\Omega^{(k)}$ for $k\in\{0, 1, \cdots, K\}$, where $\Omega^{(0)}=\{it: t<E_i^{(1)}\}$ and $\Omega^{(k)}=\{it: E_i^{(k)} \leq t <E_i^{(k+1)} \}$ for $k\in\{1, 2 \cdots, K\}$. $\Omega^{(0)}$ therefore consists of untreated unit-time observations, $\Omega^{(1)}$ consists of unit-time observations that have experienced exactly 1 treatment event, $\Omega^{(2)}$ consists of unit-time observations that have experienced exactly 2 treatment events, and so on up to $\Omega^{(K)}$. 
	
\paragraph{Potential outcomes.} $Y_{it}(E)$ is the potential outcome of unit $i$ at time $t$ in the scenario where their treatment trajectory is given by $E$. $Y_{it}(\infty,\infty,...,\infty)$ is hence the untreated scenario, in which no events are experienced, denoted $Y_{it}^{(0)}$ for compactness. $Y_{it}\equiv Y_{it}(E_i)$ is the realised outcome. 
	
\subsection{Target Estimands}\label{sect:estimation_target} 
	
The potential outcome framework described above can be used to define several distinct notions of `treatment effect'. First, the total effect of a treatment trajectory up to time $t$ for unit $i$ is given by 
\begin{equation}\label{eq:totaleffect}
    \tau_{it}=Y_{it}-Y_{it}^{(0)}.  
\end{equation}
For $it\in\Omega^{(1)}$, $\tau_{it}$ reflects the effect of the first treatment event only. For $it\in\Omega^{(k)}$ $k>1$, by contrast, it reflects the dynamic effect of all events experienced by unit $i$ in periods up to and including $t$. 

I assume researchers are interested in decomposing the total causal effect of a treatment trajectory into portions attributable to different event occurrences. I focus specifically on the incremental treatment effect of experiencing $k$ events relative to experiencing $k-1$. Such effects can be formalised by considering a sequence of counterfactual treatment trajectories, $\tilde{E}^{(k)}_i$, which take the value of $E_i$ for the first $k$ elements and $\infty$ for the remaining $K-k$ elements. Let $Y_{it}^{(k)} \equiv Y_{it}(\tilde{E}^{(k)}_i)$, denote the potential outcome for unit $i$ at time $t$ in the counterfactual scenario where they experience only the first $k$ events of their total treatment trajectory. Ruling out anticipation effects, as we do below, it makes sense to define such counterfactuals only for $k\leq K_{it}$ where $K_{it}$ is the number of events that unit $i$ has actually experienced at time $t$ ($K_{it}=\sum_{k=1}^{K}\mathbf{1}[E_i^{(k)}\leq t]$). 
	
With these concepts in hand, we can define the incremental effect of event $k$ as the impact of event $k$ on observation $(i,t)$ relative to the counterfactual where the $k$-th event did not occur but the treatment trajectory for the first $k-1$ events was unchanged:
\begin{equation}\label{eq:marginaleffect}
    \tau^{(k)}_{it}=Y_{it}^{(k)}-Y_{it}^{(k-1)} \quad \textrm{for} \quad k\in\{1,\cdots,K_{it}\}.    
\end{equation}
Noting that $Y_{it}^{(K_{it})}=Y_{it}$, we therefore have $$\sum_{k=1}^{K_{it}}\tau^{(k)}_{it}=Y_{it}^{(K_{it})}-Y^{(0)}_{it}=Y_{it}-Y^{(0)}_{it}=\tau_{it}.$$

While the definition of $\tau^{(k)}_{it}$ thus decomposes the total causal effect of a treatment trajectory into additive event-specific components, this does not necessarily impose additivity on the effects of the events themselves. It is perfectly possible that the total effect of event $k$ is determined by the effect of event $k-1$. The impact of a drought on agricultural yields, for example, may depend on the impact of droughts experienced in previous periods. The definition of counterfactuals $Y_{it}^{(k)}$ admits such interaction effects and the event-specific $\tau^{(k)}_{it}$ should therefore be interpreted as the impact of the $k$-th event occurring after the sequence $\tilde{E}^{(k-1)}_i$, including any interaction with events occurring earlier in the sequence.
	
It should be recognised that the incremental effect $\tau^{(k)}_{it}$ defined with respect to the counterfactual treatment trajectories $\tilde{E}^{(k)}_i$ and $\tilde{E}^{(k-1)}_i$ may not be the only treatment effect of interest. An alternative would be to define the treatment effect of the $k$-th event relative to a counterfactual where the $k$-th event was not experienced but all past \emph{and future} event occurrences were kept unchanged. This would amount to a counterfactual treatment trajectory where the $k$-th element of $E_i$ was replaced by the $k+1$-th element of $E_i$, the $k+1$-th element was replaced by the $k+2$-th element and so on, until finally the last finite element of $E_i$ was replaced by $\infty$. While a legitimate object of interest, such treatment effects cannot be used to decompose the total effect of the treatment trajectory into additive, event-specific components apart from in the specific case where there are no interactions between event effects. They also have a less straightforward interpretation when observations experience a different number of events. I therefore focus on $\tau^{(k)}_{it}$, which represents the incremental effect of experiencing $k$ events relative to experiencing $k-1$ for unit $i$ at time $t$. 
	
I next follow \citet{borusyak_revisiting_2024} and assume the objects of interest are scalar weighted sums of the incremental, event-specific treatment effects: $\tau^{(k)}_{w}=\sum_{it\in\Omega^{(\geq k)}}w_{it}\tau^{(k)}_{it}$, where $\Omega^{(\geq k)}\equiv\bigcup_{j\geq k}\Omega^{(j)}$ collects the observations on which the $k$-th event can affect outcomes. The weights $w_{it}$ are pre-specified by the researcher and, conditional on realised treatment timing, are non-stochastic and independent of realised outcomes. To fix ideas, I focus on the ATT of the $k$-th event at a fixed horizon $h$ after it occurs. Writing $\mathcal{N}^{(k)}_h=\bigl\{i:E^{(k)}_i<\infty,\;(i,E^{(k)}_i+h)\in\Omega\bigr\}$ for the set of units observed at horizon $h$ after their $k$-th event, this corresponds to $w_{it}=\mathbf{1}[\ell_{it}^{(k)}=h]/\bigl|\mathcal{N}^{(k)}_h\bigr|$. Extending the ATT to condition on discrete observable characteristics is straightforward and amounts to restricting the sum to observations within the discrete group of interest. 

\subsection{Continuous Treatment} 
It is straightforward to extend the concepts discussed above to the context where each treatment event occurrence is associated with a specific treatment intensity. The main amendment necessary is to define the treatment trajectory $E_i$ as a sequence of tuples \newline $E_i=\{\{E^{(1)}_i,I^{(1)}_i\},\{E^{(2)}_i,I^{(2)}_i\},\cdots,\{E^{(K)}_i,I^{(K)}_i\}\}$, which records the intensity of the $k$-th event occurrence $I^{(k)}_i>0$ alongside the period in which it occurred (with $I^{(k)}_i=0$ if unit $i$ experiences fewer than $k$ events). The counterfactual trajectories $\tilde{E}^{(k)}_i$ used to define $\tau^{(k)}_{it}$ are defined similarly to the case of binary treatment, but additionally set all intensity measures $I^{(j)}_i=0$ for $j>k$ while keeping $I^{(j)}_i$ unchanged from $E_i$ for $j \leq k$. 
	
In this context, the researcher may reasonably wish to standardise treatment effects by intensity and focus on $\tau^{(k)}_{w}=\sum_{(i,t):\ell_{it}^{(k)}\geq 0}w_{it}(\tau^{(k)}_{it}/I^{(k)}_{i})$. As long as $I^{(k)}_{i}$ is observable, this corresponds to a specific choice of weights $w'_{it}=w_{it}/I^{(k)}_{i}$. 
	
\section{Identification Limits of Conventional Difference-in-Differences Assumptions in Recurrent-Event Settings}\label{sect:wooldridge}
	
This section examines the limits of identification in recurrent-event settings when the only assumptions imposed are slight modifications on those conventionally applied in single-event settings. To ease comparison with existing literature, I follow \citet{wooldridge_twoway_2025} in the naming and general exposition of these assumptions, which can be recast into the recurrent-event framework as follows.\footnote{While the \ref{ass:sutva} assumption is a straightforward repetition, the remaining assumptions differ slightly from their single-event counterpart. The assumption labels used in this paper should therefore be understood as referring to the recurrent-event statement of the corresponding assumption.}

\begin{assumption}[SUTVA]{Stable Unit Treatment Value Assumption}\assumptionlabel{ass:sutva}{SUTVA} The potential outcome of each unit in the population is independent of the treatment timing of any other unit in the population. 
\end{assumption}
    
In ruling out spillover effects, the validity of SUTVA must be evaluated carefully in any setting. 
 
\begin{assumption}[NA]{No Anticipation}\assumptionlabel{ass:anticipation}{NA}
    $Y_{it}=Y_{it}^{(k)}$ for all $it \in \Omega^{(k)}$ and for each $k\in\{0,1,...,K\}$.
\end{assumption}
\ref{ass:anticipation} in a single-event setting requires that outcomes in periods before the treatment occurs are not affected by treatment. The intuitive extension of this to recurrent-event settings is that each treatment event occurrence can only affect outcomes after it has occurred. It follows that the sample on which the $k$-th event occurrence can affect observed outcomes is given by $\bigcup_{j=k}^{K} \Omega^{(j)}$.
	
All difference-in-difference estimators are underpinned by a `parallel trend' assumption that places structure on how untreated outcomes evolve over time. This structure facilitates difference-in-difference estimators by effectively identifying unbiased and observable counterfactual untreated outcomes for treated observations. While it is common to define parallel trends unconditionally, a more flexible, conditional version allows the time-path of untreated outcomes to depend on covariates. For treatment effects to remain identified by conventional estimators under the conditional parallel trends assumption, one needs to ensure that the covariates on which parallel trends are conditioned are unaffected by treatment. In the recurrent-event setting, this requires that covariates are unaffected by the occurrence of \emph{any} event, which is more plausible when one restricts $\mathbf{x}$ to time-invariant characteristics.
	
\begin{assumption}[NBC]{No Bad Controls}\assumptionlabel{ass:nbc}{NBC} Letting $\mathbf{x}(E_i)$ denote time-invariant covariates when the treatment trajectory is $E_i \in \mathcal{E}$, assume $\mathbf{x}(E_i)=\mathbf{x}(\infty), \forall \, E_i \in \mathcal{E}$.
\end{assumption}
	
To formalise the recurrent-event conditional parallel trends assumption, let $|\mathcal{E}| \equiv G$ denote the number of all possible timing trajectories.\footnote{The dimension of $\mathcal{E}$ follows from the observation that the number of possible timing configurations for experiencing $k$ events amounts to choosing $k$ periods from those in which events may occur. Assuming no simultaneous events, such choice is done without replacement. In a panel of length $T$ in which every unit is observed for at least $P$ periods before any event occurs and in which each unit experiences at most $K$ events, the total number of possible trajectory groups is therefore $|\mathcal{E}| \equiv G = \sum_{m=0}^{K}\binom{T-P}{m}$. In practice, realised panels will likely contain fewer timing trajectories than all that are feasible, in which case $G$ denotes the observed count of distinct timing trajectories.} For each trajectory group $e\in\mathcal{E}$, define a dummy $de$ equal to $1$ for units with $E_i=e$ and $0$ otherwise, and collect these into the vector $\mathbf{d}=(d1,\ldots,dG)$. While difference-in-differences can account for initial level differences in outcomes between trajectory groups, conditional parallel trends in the recurrent-event setting stipulates that the change in untreated outcomes over time cannot depend on $\mathbf{d}$.
	
\begin{assumption}[CPT]{Conditional Parallel Trends}\assumptionlabel{ass:cpt}{CPT} For $t=2,...,T$ and time-constant controls $\mathbf{x}$,
    \begin{equation}
        \E[Y^{(0)}_{it}-Y^{(0)}_{i1}\mid\mathbf{d},\mathbf{x}]     =\E[Y^{(0)}_{it}-Y^{(0)}_{i1}\mid\mathbf{x}]
    \end{equation}
\end{assumption}
While \ref{ass:cpt} thus imposes that trends cannot differ according to treatment event timing, they can differ according to covariates $\mathbf{x}$, which may, in turn, correlate with a unit's treatment trajectory. If units are observed for at least two periods prior to treatment, one can relax CPT to allow untreated outcomes to include a linear trajectory-specific time trend (i.e. $\gamma_e de t$ for each $e\in \mathcal{E}$) \citep{wooldridge_twoway_2025}. While this may be attractive in some settings, it makes notation somewhat more cumbersome without changing the main conclusions of this section and I therefore focus on the more restrictive \ref{ass:cpt}.

A final assumption obtains estimating equations amenable to OLS regression by stipulating that the relationship between untreated outcomes and covariates $\mathbf{x}$ is linear in parameters. To formalise this condition, let $fs_t$ for $s=1,...,T$ denote time period dummies so that $fs_t=1$ if $s=t$ and zero otherwise. 
\begin{assumption}[LIN]{Linearity}\assumptionlabel{ass:lin}{LIN}
    For trajectory group indicators $de$, $e\in\mathcal{E}$, and time-constant controls $\mathbf{x}$,
    \begin{align}
        \E[Y_{i1}^{(0)}|\mathbf{d},\mathbf{x}] &= \alpha + \sum_{e\in\mathcal{E}}\beta_e de_i + \mathbf{x}_i'\kappa + \sum_{e\in\mathcal{E}}de_i\,\mathbf{x}_i'\xi_e \label{eq:lin-t1} \\
        \E[Y_{it}^{(0)}|\mathbf{d},\mathbf{x}]-\E[Y_{i1}^{(0)}|\mathbf{d},\mathbf{x}] &= \sum_{s=2}^T\gamma_s fs_t + \sum_{s=2}^T fs_t\,\mathbf{x}_i'\pi_s,\quad t=2,\ldots,T \label{eq:lin-trend}
    \end{align}
\end{assumption}
Since $\mathbf{d}$ does not appear on the right-hand side of equation \eqref{eq:lin-trend}, \ref{ass:lin} implies \ref{ass:cpt} and the latter assumption is thus superfluous once the former is imposed.
	
To examine the limits of identification in a recurrent-event setting under assumptions \ref{ass:sutva}, \ref{ass:anticipation}, \ref{ass:nbc} and \ref{ass:lin}, it is instructive to first examine what quantities are identified. To this end, consider the following model of observed outcomes
\begin{equation}\label{eq:dgp_marginal}
    Y_{it}=\sum_{k=1}^{K}\sum_{s=E_i^{(k)}}^{T}\tau^{(k)}_{is}\left(w^{(k)}_{it}\cdot fs_t\right)+\sum_{s=2}^{T}\gamma_s fs_t+\sum_{s=2}^T fs_t\,\mathbf{x}_i'\pi_s+c_i+u_{it},
\end{equation}
where $\tau^{(k)}_{is}$ are the event-specific effects defined by equation \eqref{eq:marginaleffect}, $w^{(k)}_{it}\equiv\mathbf{1}[t \geq E_i^{(k)}]$ and $u_{it}$ is an unobserved disturbance. $c_i$ is a unit-specific fixed effect that adds a time-invariant, unobserved unit component $\upsilon_i$ to the terms of equation \eqref{eq:lin-t1} 
\begin{equation}\label{eq:unitfe}
    c_i=\alpha + \sum_{e\in\mathcal{E}}\beta_e de_i + \mathbf{x}_i'\kappa + \sum_{e\in\mathcal{E}}de_i\,\mathbf{x}_i'\xi_e + \upsilon_i.
\end{equation}
Under \ref{ass:lin}, the non-stochastic terms in \eqref{eq:unitfe} together with the time effects in \eqref{eq:dgp_marginal} constitute the conditional expectation of untreated outcomes given $(\mathbf{d},\mathbf{x})$. The remaining components, $\upsilon_i$ and $u_{it}$, are deviations from that expectation and therefore satisfy $\E[\upsilon_i|\mathbf{d},\mathbf{x}]=0$ and $\E[u_{it}|\mathbf{d},\mathbf{x}]=0$ for $t=1,\ldots,T$. Since $\mathbf{d}$ records the full timing trajectory, any cell defined by event occurrence and horizon is a function of $\mathbf{d}$ and $t$, and the composite disturbance is therefore mean-zero within every such cell. 
	
The $\left(w^{(k)}_{it}\cdot fs_t\right)$ interactions in \eqref{eq:dgp_marginal} emulate the exposition of \citet{wooldridge_twoway_2025} and ensure that only effects dated $t$ for events that have occurred by period $t$ contribute to the double sum component. Following the definition of the event-specific treatment effects, we can therefore rewrite \eqref{eq:dgp_marginal} replacing the double sum over event-specific effects with a single product of the total trajectory effect $\tau_{it}$ and the indicators $w^{(1)}_{it}$ that select periods after the first event have occurred
\begin{equation}\label{eq:dgp_total}
    Y_{it}=\tau_{it}w^{(1)}_{it}+\sum_{s=2}^{T}\gamma_s fs_t+\sum_{s=2}^T fs_t\,\mathbf{x}_i'\pi_s+c_i+u_{it}.
\end{equation}
The value of this rewriting is that it demonstrates how a data generating process in which observed outcomes are affected by the dynamic effects of several events can be rewritten as one that includes a single treatment effect parameter per treated observation. Such a linear-in-parameters DGP featuring an unobserved unit effect and a single treatment effect parameter per treated observation is the workhorse model deployed in the single-event event-study literature. It follows that very slight modifications of existing estimators developed by this literature are capable of recovering weighted averages of the observation-specific $\tau_{it}$ parameters -- for any set of deterministic, researcher-specified weights -- under the recurrent-event extensions of the conventional single-event assumptions outlined above. In single-event settings, these average effects amount to ATTs of the single treatment event. When events can reoccur, by contrast, they capture the dynamic impact of all distinct events affecting the observation i.e. they are average total treatment trajectory effects.
	
The imputation estimator of \citet{borusyak_revisiting_2024} (BJS) offers a clean demonstration of this observation. Just as in the single-event setting, this estimator first estimates the parameters of the untreated potential outcome model (i.e. the $\gamma_s$ and $\pi_s$ from equation \eqref{eq:dgp_total} and the $\alpha, \,\beta_e, \, \kappa,\, \xi_e$ and $\upsilon_i$ of equation \eqref{eq:unitfe}) via OLS regression on the untreated sample $\Omega^{(0)}$. The \ref{ass:sutva}, \ref{ass:anticipation} and \ref{ass:lin} assumptions ensure that the estimates from this subsample regression are unbiased for the true model parameters conditional on $(\mathbf{d},\mathbf{x})$. Combined with assumption \ref{ass:nbc}, the predicted values from this first-stage regression serve as a valid estimator for counterfactual untreated outcomes $\widehat{Y}^{(0)}_{it}$. The residuals from this projection $Y_{it}-\widehat{Y}^{(0)}_{it}$ thus amount to observation-specific estimates of $\tau_{it}+u_{it}$ for the remaining portion of the sample $\bigcup_{j=1}^{K} \Omega^{(j)}$. Averaging these projection residuals over a sufficiently large subsample eliminates the idiosyncratic terms $u_{it}$ and thereby isolates the average total treatment trajectory effect among the subsample.
	
While the BJS imputation estimator provides a particularly clear example of how one can recover average total treatment trajectory estimates when events can reoccur, it is also straightforward to extend the equivalence result of \citet{wooldridge_twoway_2025} to the recurrent-event settings. This result follows directly from Wooldridge's proofs by replacing the cohort indicators of his exposition with the trajectory indicators $de$ and by imposing the recurrent-event assumptions \ref{ass:sutva}, \ref{ass:anticipation}, \ref{ass:nbc} and \ref{ass:lin}, instead of their single-event analogues. In the recurrent-event setting, the equivalence implies that average effects of a given treatment trajectory (alongside any combination or conditional expectation of these quantities) can be obtained by several distinct but equivalent estimators.\footnote{As in \citet{wooldridge_twoway_2025}, the equivalence is stated for a balanced panel. Where the panel is unbalanced and the untreated component is parametrised by unit fixed effects rather than trajectory dummies, the estimators remain asymptotically equivalent but need not coincide in finite samples.} 

The equivalent estimator set includes an `extended TWFE' model in which the event-time dummies of conventional specifications are replaced with interactions between the timing trajectory dummies and the period dummies over treated cells. Where the untreated model contains covariates, one additionally interacts the trajectory-period interactions with the covariate vector $\mathbf{x}$, centered around the trajectory-period cell average. The coefficient on the trajectory-period interaction then has interpretation as the average total trajectory effect within the trajectory-period cell. This highlights that, just as in single-event settings, TWFE models are capable of returning interpretable parameter estimates in recurrent-event settings provided they are sufficiently flexible. At the time of writing, however, no TWFE model deployed in a recurrent-event setting  includes the terms that render the model `sufficiently flexible' in the \citet{wooldridge_twoway_2025} sense. This point is returned to in Section \ref{sect:twfe}.
	
A final point to note here is the parallel between the results of this subsection and those of \citet{de_chaisemartin_did_2024}. In their more general setting, \citeauthor{de_chaisemartin_did_2024} define the `non-normalized actual-versus-status-quo' effect as the difference between the realized outcome and a counterfactual `status-quo' outcome in which unit's treatment status is unchanged from its initial value. In the present context, the assumption that $E_i^{(1)}>1$ implies the status-quo outcome of all units in the sample is the untreated outcome $Y_{it}^{(0)}$. The total trajectory treatment effect is therefore identical to \citeauthor{de_chaisemartin_did_2024}'s non-normalized actual-versus-status-quo effect. Relative to their work, a contribution of this section is demonstrating how the framework of \citet{wooldridge_twoway_2025} can be generalised to settings featuring recurrent events. This is more flexible than the \citeauthor{de_chaisemartin_did_2024} exposition, in that it allows for covariate-specific time trends, and also highlights that \citeauthor{wooldridge_twoway_2025}'s equivalence result applies in recurrent-event settings. That the same can be achieved in the more general environment of \citet{de_chaisemartin_did_2024} is a conjecture I leave for future work.
	
For observations $it\in\Omega^{(1)}$, which have experienced exactly one event, the total trajectory effect is simply the effect of the first event. For observations that have experienced two or more events, by contrast, the total trajectory effect is a single parameter reflecting the combined effect of all events experienced by unit $i$ up to and including time $t$. For observations $it\in\Omega^{(2)}$, for example, $\tau_{it}=\tau^{(1)}_{it}+\tau^{(2)}_{it}$. Without further restrictions, there are infinite possible ways in which $\tau_{it}$ can be decomposed into the event-specific components $\tau^{(1)}_{it}$ and $\tau^{(2)}_{it}$ and assumptions \ref{ass:sutva}, \ref{ass:anticipation}, \ref{ass:nbc} and \ref{ass:lin} provide no basis for identifying the true decomposition. A data generating process satisfying \ref{ass:sutva}, \ref{ass:anticipation}, \ref{ass:nbc} and \ref{ass:lin}, for example, would result in the same distribution of observables as another constructed by adding a constant $c$ to $\tau^{(1)}_{it}$ and subtracting it from $\tau^{(2)}_{it}$ at every observation in $\Omega^{(2)}$. Since the assumptions place no restriction on event-1 effects across the $\Omega^{(1)}$ and $\Omega^{(2)}$ subsamples, both processes satisfy them, and any estimator that is consistent for the total trajectory effect would return the same value across the distinct processes. 

Thus, while average total trajectory effects are identified under assumptions similar to those conventionally applied in single-event settings, the event-specific effects $\tau_{it}^{(k)}$ are not, in general.
	
\section{Identifying Event-Specific Effects}\label{sect:estimator}
	
\subsection{Identifying Assumptions}
To identify event-specific effects, it is necessary to place some restriction on their form. While several assumptions can yield identification, the following is proposed as a relatively general restriction that mirrors the familiar `conditional parallel trends' assumption introduced above. 
	
\begin{assumption}[CPET]{Conditional Parallel Effect-Trends.}\assumptionlabel{ass:cpet}{CPET}
    For each $k\in\{2,\ldots,K\}$, each $m\in\{1,\ldots,k-1\}$, and each prior-event
    relative time $\ell\geq1$:
    \[
    \E\left[\tau^{(m)}_{i\ell}-\tau^{(m)}_{i0} \;\middle|\;
    (i,E_i^{(m)}+\ell)\in \Omega^{(k)},\, \mathbf{x}^{(m)}_i\right]
    =
    \E\left[\tau^{(m)}_{i\ell}-\tau^{(m)}_{i0} \;\middle|\;
    (i,E_i^{(m)}+\ell)\in \Omega^{(m)},\, \mathbf{x}^{(m)}_i\right],
    \]
    where $\mathbf{x}^{(m)}_i$ is a function of the time-invariant covariates $\mathbf{x}$ and of the first $m$ components of the treatment trajectory.
\end{assumption}
	
\ref{ass:cpet} is similar in form to \ref{ass:cpt} but imposes structure on growth of the event-$m$ effect $\tau^{(m)}_{it}$, rather than on growth of the untreated potential outcome $Y_{it}^{(0)}$. Similarly to \ref{ass:cpt}, the conditioning vector $\mathbf{x}^{(m)}_i$ cannot contain characteristics that either relate to subsequent events or that are affected by them. Being a function of the trajectory up to event $m$ and of $\mathbf{x}$, the conditioning vector nonetheless accommodates attributes of the $m$-th event such as measures of its intensity and the number of periods between event $m$ and the previous event (i.e. $E^{(m)}_i-E^{(m-1)}_i$). Since $\mathbf{x}^{(m)}_i$ effectively captures heterogeneity in event effects, it is advisable to include any attribute that the researcher wishes to examine as a dimension of such heterogeneity.
	
To judge the plausibility of \ref{ass:cpet}, it is helpful to view it as ruling out a particular type of selection. Specifically, it rules out that growth in the event-$m$ effect governs the selection of units into having experienced another event by period $E_i^{(m)}+\ell$. If the event of interest were job loss and the outcome of interest were earnings, for example, \ref{ass:cpet} would be violated if individuals who experienced persistent earnings reductions after a first job loss event (conditional on characteristics $\mathbf{x}^{(m)}_i$) were more likely to lose another job than individuals who experienced only a short-lived earnings drop. While it is not too hard to think of mechanisms through which such selection could occur in this example, \ref{ass:cpet} seems more plausible when the events of interest are due to factors outside the immediate control of economic agents, such as natural disasters or episodes of severe air pollution. Even in settings where \ref{ass:cpet} seems a relatively strong assumption, supporting evidence can be obtained by comparing estimated effect paths among units that do and do not go on to experience a subsequent event, over horizons at which both groups remain in $\Omega^{(m)}$. Such a heuristic test is similar in spirit to the `pre-trend' analysis that is typically presented as validation of \ref{ass:cpt} and is discussed in greater detail in Section \ref{sect:test_eventrep}.
	
Just as \ref{ass:cpt} identifies treatment effects in single-event settings by effectively providing a valid counterfactual for the untreated outcome, \ref{ass:cpet} identifies event-specific treatment effects in recurrent-event settings by providing a valid counterfactual for the prior-event effects. When it holds, units that have experienced exactly $m$ events at relative time $E^{(m)}_i+\ell$ provide a valid counterfactual for growth in event-$m$ effects for units who have experienced another event by relative time $E^{(m)}_i+\ell$. As highlighted above, \ref{ass:sutva}, \ref{ass:anticipation}, \ref{ass:nbc} and \ref{ass:lin} identify the first-event effects $\tau^{(1)}_{i\ell}$ for observations $it \in \Omega^{(1)}$. \ref{ass:cpet} then enables an inductive process by providing valid counterfactuals that first identify second-event effects, then third-event effects and so on until the effects of all events up to event $K$ have been recovered. For this to be possible, however, the sample must include the data necessary for all required counterfactuals to be estimable. Formally, this requires 
	
\begin{assumption}[SUPP]{Support for CPET imputation.}\assumptionlabel{ass:supp}{SUPP}
    For each $k\in\{2,\ldots,K\}$, each $m<k$, and each relative time $\ell\geq1$ at
    which $\tau^{(m)}_{i\ell}$ is required for units in $\Omega^{(k)}$,
    \[
    \Pr\bigl[(i,E_i^{(m)}+\ell)\in\Omega^{(m)}\bigr]>0
    \]
    and $\E\bigl[\mathbf{x}^{(m)}_i\mathbf{x}^{(m)\prime}_i\mid
    (i,E_i^{(m)}+\ell)\in\Omega^{(m)}\bigr]$ is nonsingular.
\end{assumption}
	
\ref{ass:supp} ensures that the comparison group required by \ref{ass:cpet} is populated at every relative time the counterfactual construction needs and exhibits sufficient variation in $\mathbf{x}^{(m)}_i$, so that the growth relationship can be estimated at each such horizon. To give a concrete example of the first condition, suppose that $K=2$ and $T=4$. If the sample includes units that experience a first event in period 2 and a second event in period 3, then identification of the second event's impact effect, $\tau^{(2)}_{i\ell}$ at horizon $\ell=0$, requires other units in the sample that experience a first event in period 2 and a second event in period 4 at the earliest. Similarly, identification of the second event's effect at horizon $\ell=1$ requires the sample to feature units that experience a first event in period 2 and no second event. In practice, \ref{ass:supp} may hold for low event counts (e.g. $k=2$) but may be violated for later events (e.g. $k=3$), as identification of later events requires sample support over longer prior-event horizons $\ell$. Because identification follows inductive logic, early event effects can be identified even if later events are not. 

Just as \ref{ass:lin} enables straightforward estimation of total trajectory effects, estimating event-specific effects is facilitated by imposing further restriction on the relationship between effect growth and $\mathbf{x}^{(m)}_i$. Following \ref{ass:lin}, this relationship is specified as linear in parameters, which allows prior-event counterfactuals to be estimated by OLS. 
	
\begin{assumption}[LEG]{Linear Effect Growth.}\assumptionlabel{ass:leg}{LEG}
    For each $m\in\{1,\ldots,K-1\}$ and each relative time $\ell\geq1$, there exists
    a parameter vector $\lambda^{(m)}_\ell$ such that, for all units with
    $E^{(m)}_i<\infty$,
    \[
    \E\bigl[\tau^{(m)}_{i\ell}-\tau^{(m)}_{i0}\;\big|\;\mathbf{d},\mathbf{x}\bigr]
    =\mathbf{x}^{(m)\prime}_i\lambda^{(m)}_\ell,
    \]
    where $\mathbf{x}^{(m)}_i$ is as defined in Assumption~\ref{ass:cpet} and
    contains a constant.
\end{assumption}
	
Since $\mathbf{d}$ records the full timing trajectory, whether a unit lies in $\Omega^{(m)}$ or in $\Omega^{(k)}$ at horizon $\ell$ is a function of $\mathbf{d}$. Taking the expectation of \ref{ass:leg} conditional on either the $\Omega^{(m)}$ or the $\Omega^{(k)}$ subsample together with $\mathbf{x}^{(m)}_i$ therefore leaves the right-hand side unchanged, since it is a function of $\mathbf{x}^{(m)}_i$ alone, implying that both conditional means in \ref{ass:cpet} equal $\mathbf{x}^{(m)\prime}_i\lambda^{(m)}_\ell$. \ref{ass:leg} thus implies \ref{ass:cpet} and, as with \ref{ass:lin} and \ref{ass:cpt}, the latter is superfluous once the former is imposed.
	
While \ref{ass:leg} may appear restrictive, the model is linear in parameters rather than in the underlying characteristics and can accommodate a considerable degree of flexibility. Where the elements of $\mathbf{x}^{(m)}_i$ are discrete and the model saturated, for example, the linear form is without loss of generality, since any conditional mean function of a discrete variable can be written as a linear combination of indicators for its support points. In that case $\mathbf{x}^{(m)\prime}_i\lambda^{(m)}_\ell$ is the mean growth in the event-$m$ effect within the cell to which unit $i$ belongs, and the substantive content of \ref{ass:leg} is then the requirement that this cell mean does not vary with the remaining components of $(\mathbf{d},\mathbf{x})$, in particular with the subsequent trajectory. Where the elements of $\mathbf{x}^{(m)}_i$ are continuous, the linear form becomes a genuine restriction, although non-linear relationships can be accommodated by including transformations of the underlying variables or by modelling the continuous components using splines. The value of the linear form is that it permits the fitted relationship to be evaluated at values of $\mathbf{x}^{(m)}_i$ not observed among units remaining in $\Omega^{(m)}$, which is why \ref{ass:supp} imposes no condition on the overlap of $\mathbf{x}^{(m)}_i$ between the two groups. 
	
\subsection{The Sequential Imputation Estimator}
\ref{ass:supp} and \ref{ass:leg}, together with \ref{ass:sutva}, \ref{ass:anticipation}, \ref{ass:nbc} and \ref{ass:lin}, identify event-specific effects and justify the following sequential imputation estimator (SIE), whose consistency is established below under further regularity conditions. In what follows, $\mathbf{x}_i$ denotes the time-invariant covariates entering \ref{ass:lin} and $\mathbf{x}^{(m)}_i$ the event-$m$ growth covariates entering \ref{ass:cpet} and \ref{ass:leg}, which are constructed from $\mathbf{x}_i$ and the trajectory through event $m$. 
	
\begin{definition}[Sequential Imputation Estimator]\label{def:sie}\leavevmode
    \setlength{\abovedisplayskip}{5pt}%
    \setlength{\belowdisplayskip}{5pt}%
    \begin{description}[leftmargin=0pt,style=unboxed,itemsep=0.8ex,topsep=0.8ex,parsep=0pt]
        
        \item[Stage 0.] Estimate $\widehat c_i$, $\widehat\gamma_t$ and $\widehat\pi_t$ of model \eqref{eq:dgp_marginal} via OLS of $Y_{it}$ on unit indicators, period indicators and interactions of the period indicators with $\mathbf{x}_i$ using $\Omega^{(0)}$, giving $\widehat{Y}^{(0)}_{it}=\widehat c_i+\widehat\gamma_t+\mathbf{x}_i'\widehat\pi_t$ for all $it$. Total trajectory effects for all $it\in\bigcup_{m=1}^K\Omega^{(m)}$ are then estimated as $$\widehat\tau_{it}=Y_{it}-\widehat{Y}^{(0)}_{it}.$$ 
        
        \item[Stage $\boldsymbol{m}$, for $\boldsymbol{m=1,\ldots,K-1}$:]\mbox{}
        \begin{description}[style=unboxed,leftmargin=1.5em,itemindent=0pt,itemsep=0.6ex,topsep=0.6ex,parsep=0pt]
            
            \item[$\boldsymbol{m}$.\textbf{a.}] For all $it\in\Omega^{(m)}$, estimate $\tau^{(m)}_{it}$ as
            \[
            \widehat\tau^{(m)}_{it}=\widehat\tau_{it}
            -\sum_{j=1}^{m-1}\widehat\tau^{(j)}_{it},
            \]
            where $\sum_{j=1}^{0}\cdot\equiv0$.
            
            \item[$\boldsymbol{m}$.\textbf{b.}] Following \ref{ass:leg}, model growth in the event-$m$ effect at each relative time $\ell\geq1$ as
            \[
            \tau^{(m)}_{i,E^{(m)}_i+\ell}-\tau^{(m)}_{i,E^{(m)}_i}
            =\mathbf{x}^{(m)\prime}_i\lambda^{(m)}_\ell,
            \]
            and estimate $\widehat{\lambda}^{(m)}_\ell$ via OLS of $\bigl(\widehat\tau^{(m)}_{i,E^{(m)}_i+\ell}-\widehat\tau^{(m)}_{i,E^{(m)}_i}\bigr)$ on $\mathbf{x}^{(m)}_i$ using units $\{i:(i,E^{(m)}_i+\ell)\in\Omega^{(m)}\}$, which is nonempty by \ref{ass:supp}.
            
            \item[$\boldsymbol{m}$.\textbf{c.}] For all $it\in\Omega^{(k)}$, $k>m$, estimate $\tau^{(m)}_{it}$ as
            \[
            \widehat\tau^{(m)}_{it}=\widehat\tau^{(m)}_{i,E^{(m)}_i}
            +\mathbf{x}^{(m)\prime}_i\widehat{\lambda}^{(m)}_{t-E^{(m)}_i}.
            \]
        \end{description}
        
        \item[Stage $\boldsymbol{K}$.] For all $it\in\Omega^{(K)}$, estimate $\tau^{(K)}_{it}$ as
        \[
        \widehat\tau^{(K)}_{it}=\widehat\tau_{it}
        -\sum_{j=1}^{K-1}\widehat\tau^{(j)}_{it}.
        \]
        
        \item[Aggregation.] Aggregate the observation-specific $\widehat\tau^{(m)}_{it}$ using pre-specified weights $w_{it}$ to obtain the quantity of interest. The average effect of event $m$ at horizon $\ell$, for example, is estimated as
        \[
        \widehat\tau^{(m)}_\ell=\bigl|\mathcal{N}^{(m)}_\ell\bigr|^{-1}
        \sum_{i\in\mathcal{N}^{(m)}_\ell}\widehat\tau^{(m)}_{i\ell},
        \]
        where $\widehat\tau^{(m)}_{i\ell}$ denotes $\widehat\tau^{(m)}_{it}$ at $t=E^{(m)}_i+\ell$.
    \end{description}
\end{definition}
	
Several aspects of the SIE deserve comment. First, since the environment assumes events are not simultaneous (i.e. $E^{(m)}_i<E^{(m+1)}_i$ for all events $m$ that unit $i$ eventually experiences), it is certain that $(i,E^{(m)}_i)\in\Omega^{(m)}$ whenever unit $i$ experiences an $m$-th event. This ensures the event-$m$ impact effect $\widehat\tau^{(m)}_{i,E^{(m)}_i}$ for units that experience an $m$-th event is always obtained in Stage $m$.a. This is why \ref{ass:cpet} and \ref{ass:leg}, which underpin steps $m$.b and $m$.c, restrict \emph{growth} in event-specific treatment effects rather than their level. Each unit's imputed event-specific effect path is anchored at its own realised impact effect rather than at a group average, so heterogeneity in the level of event effects is accommodated without restriction and \ref{ass:cpet} is called upon only to extend that level to horizons by which point the unit has experienced a subsequent event.
	
Second, the regressand at Stage $m$.b is constructed from estimated effects rather than observed data, and therefore carries estimation error alongside any departure of the true growth relationship from \ref{ass:leg}. For OLS to recover $\lambda^{(m)}_\ell$, and so for Stage $m$.c to construct valid counterfactuals, the regressors must be orthogonal to both prior-event estimation error and departures from linearity. \ref{ass:lin} and \ref{ass:leg} together imply these two quantities are mean zero conditional on $(\mathbf{d},\mathbf{x})$ which, since $\mathbf{x}^{(m)}_i$ is a function of the treatment trajectory and the time-invariant covariates, ensures orthogonality. \ref{ass:cpet} then justifies extrapolation of the fitted relationship beyond the sample on which it is estimated, by asserting that growth estimated among units remaining in $\Omega^{(m)}$ also describes those that have moved on to $\Omega^{(k)}$.
	
Finally, the estimator is sequential in the sense that consistency at stage $m$ presupposes consistency at every preceding stage, since Stage $m$.a subtracts the prior-event effects constructed in Stage $j$.c for each $j<m$. \ref{ass:cpet}, \ref{ass:supp} and \ref{ass:leg} must therefore hold at every $j<m$, and a failure at any stage propagates to all subsequent ones. Where \ref{ass:supp} fails at a given $(m,\ell)$, for example, any average effect which requires the unavailable $\lambda^{(m)}_\ell$ should be reported as inestimable.\footnote{An alternative that constructs the estimate omitting the units requiring $\lambda^{(m)}_\ell$ would implicitly condition on a subpopulation determined by the inter-event gap, which $\mathbf{x}^{(m)}_i$ permits to modify effects, and would hence depart from the population ATT.} The converse also holds, however, in that no stage draws on quantities constructed at later ones. Estimates of early event effects are therefore unaffected by failures arising at later events, so that a researcher whose sample cannot support identification of, say, third-event effects may nonetheless report first- and second-event effects.
	
To establish consistency of the SIE, several further regularity conditions are imposed.
\begin{assumption}[REG]{Regularity.}\assumptionlabel{ass:reg}{REG}\mbox{}
    \begin{enumerate}[label=(\roman*),leftmargin=2em,itemsep=0.4ex,topsep=0.4ex,parsep=0pt]
        \item $\bigl\{\bigl(E_i,\mathbf{x}_i,\{Y_{it}(\cdot)\}_{t\leq T},\{u_{it}\}_{t\leq T}\bigr)\bigr\}_{i=1}^{N}$ are independent and identically distributed across $i$, with $T$ fixed and $N\to\infty$.
        \item $\E[u_{it}^2]<\infty$ for $t=1,\ldots,T$.
        \item $\E\|\mathbf{x}_i\|^2<\infty$, $\E\|\mathbf{x}^{(m)}_i\|^2<\infty$ and $\E\bigl[(\tau^{(m)}_{i\ell})^2\bigr]<\infty$ for all $m$ and $\ell$, and the products $\|\mathbf{x}_i\|^2u_{it}^2$ and $\|\mathbf{x}^{(m)}_i\|^2\bigl(\tau^{(m)}_{i\ell}\bigr)^2$ have finite expectation.
        \item Over $\Omega^{(0)}$, the period indicators $f2_t,\ldots,fT_t$ and their interactions with $\mathbf{x}_i$ are not collinear once expressed as deviations from their within-unit means.
        \item Any set of observations over which an average effect is reported constitutes a non-vanishing share of the sample as $N\to\infty$, with $\max_{it}w_{it}\to0$.\inlineqed
    \end{enumerate}
    \let\popQED\relax
\end{assumption}

\begin{proposition}[Consistency of the SIE]\label{prop:sie}
    Suppose Assumptions~\ref{ass:sutva}, \ref{ass:anticipation}, \ref{ass:nbc}, \ref{ass:lin}, \ref{ass:supp}, \ref{ass:leg} and \ref{ass:reg} hold. Then $\plim_{N\to\infty} \widehat\lambda^{(m)}_\ell=\lambda^{(m)}_\ell$ for each $m\leq K-1$ and each $\ell\geq1$, and for each $m\in\{1,\ldots,K\}$ and any weights $w_{it}$ satisfying \ref{ass:reg}(v),
    \[
    \plim_{N\to\infty}\Bigl(\sum_{it\in\Omega^{(\geq m)}}w_{it}\widehat\tau^{(m)}_{it}
    -\sum_{it\in\Omega^{(\geq m)}}w_{it}\,\E\bigl[\tau^{(m)}_{it}\mid\mathbf{d},\mathbf{x}\bigr]\Bigr)=0 .
    \]
\end{proposition}
The proof proceeds by induction on the event-count $m$ and is given in full in Appendix~\ref{app:proof_sie}. 
	
\subsection{Relation to Other Estimators}\label{sect:equivalence}
	
The SIE, as its name suggests, is an imputation estimator in the spirit of BJS. Just as in their proposed estimator, the untreated observations $it\in\Omega^{(0)}$ are used to estimate the parameters that comprise the untreated counterfactual outcome, whose extrapolation into the treated sample is justified by \ref{ass:lin}. Imputation of untreated potential outcomes in this manner is Stage~0 of the SIE, and identifies total trajectory effects. The subsequent stages implement a similar imputation procedure in which observations that have experienced exactly $m$ events are used to estimate the parameters that characterise growth in event-$m$ treatment effects, the extrapolation of which to observations that have experienced $k>m$ events is justified by \ref{ass:cpet}.
	
In light of \citet{wooldridge_twoway_2025}, who shows the BJS imputation estimator can be equivalently implemented using TWFE regression, it is worth considering whether a similar equivalence holds for the SIE. Just as Stage~0 of the SIE essentially implements the first stage of the BJS estimator to recover estimates of the total trajectory effect, estimates of trajectory-specific total effects can alternatively be obtained using OLS regression applied to sufficiently flexible TWFE models. As outlined in Section~\ref{sect:wooldridge}, such sufficiency requires interacting the trajectory dummies with period dummies and with the covariates $\mathbf{x}_i$ jointly, the latter centred within each trajectory-period cell. The regression-based approach thereby summarises the treated sample through a fixed set of cell-level parameters, unlike the imputation-based approach, which recovers observation-specific total trajectory effect estimates.
	
The remaining steps of the SIE amount to linear combinations of the Stage~0 total trajectory estimates which, under \ref{ass:leg} and \ref{ass:supp}, consistently decompose the total trajectory effects into event-specific components. To see this, observe that Stage~1.a sets the first-event estimates $\hat\tau^{(1)}_{it}$ equal to the total trajectory estimates for $it\in\Omega^{(1)}$. Stage~1.b regresses differences in these event-1 estimates on covariates and thus estimates the effect-growth as a linear combination of the total trajectory estimates for $it\in\Omega^{(1)}$. Stage~1.c then combines the growth-effect parameter estimates with the first-event impact estimates from Stage~1.a to obtain first-event effect estimates for $it\in\Omega^{(m)}$ when $m>1$. Estimates of the event-specific impacts of subsequent events proceed in the same manner, involving linear transformations of the event-specific estimates for prior events. All event-specific $\hat\tau^{(m)}_{it}$ are thereby linear combinations of the $\hat\tau_{it}$ estimates from Stage~0.
	
Linear transformations of the total trajectory effects recovered by the regression approach can likewise provide consistent estimates of average event-specific effects. As \citet{wooldridge_twoway_2025} highlights in the single-event case, it is thereby possible to obtain aggregate event-specific estimates and their standard errors from a single regression using standard software. The SIE and the regression-based approach differ, however, in how the growth covariates $\mathbf{x}^{(m)}_i$ are incorporated. Because the regression summarises the treated sample through its cell parameters and covariate slopes, the decomposition can be carried out within it only insofar as $\mathbf{x}^{(m)}_i$ is spanned by the interactions the regression contains. Where $\mathbf{x}^{(m)}_i$ is constant within trajectory groups this holds immediately, and no information is lost in implementing Stage~$m$.b with trajectory-group average effects. More generally, leveraging \ref{ass:leg} within the TWFE framework requires one of two approaches.
	
The first involves interacting the trajectory-period dummies on the treated sample with those elements of $\mathbf{x}^{(m)}_i$ that are not already spanned by the interactions in the TWFE regression. Elements that are functions of the trajectory through event $m$, including timing, inter-event gaps and intensity, are constant within trajectory groups by construction and are therefore captured by the trajectory-period dummies themselves. Since those dummies are also interacted with the centred covariates, the fitted treated effect varies linearly with $\mathbf{x}_i$ within each cell, so elements linear in $\mathbf{x}_i$, and products of these with trajectory-determined characteristics, are captured too. Further interactions are required only where \ref{ass:leg} employs non-linear transformations of $\mathbf{x}_i$. With those included, $\lambda^{(m)}_\ell$ is recovered exactly, at the cost of a slope vector for every trajectory-period cell and of a stricter support requirement than \ref{ass:supp}. Whereas \ref{ass:supp} asks only that $\E\bigl[\mathbf{x}^{(m)}_i\mathbf{x}^{(m)\prime}_i\mid(i,E_i^{(m)}+\ell)\in\Omega^{(m)}\bigr]$ be non-singular among the units remaining in $\Omega^{(m)}$ at that horizon, the interacted specification requires those elements of $\mathbf{x}^{(m)}_i$ that vary within trajectory groups to exhibit sufficient variation within each group separately, the remaining elements being absorbed by the trajectory-period dummies. Where treatment is non-binary the requirement is more demanding still, since the trajectory groups themselves must then partition on intensity as well as timing and contain correspondingly fewer units.
	
The second involves replacing $\mathbf{x}^{(m)}_i$ in \ref{ass:leg} with the trajectory-group mean characteristics. This yields a between-group estimator of $\lambda^{(m)}_\ell$, which remains consistent under \ref{ass:leg} and economises on parameters, but identifies its coefficients from variation in group mean characteristics alone. Elements of $\mathbf{x}^{(m)}_i$ determined by the trajectory up to event $m$, such as event timing, intensity or the gap since the previous event, are constant within trajectory groups by construction and hence unaffected. Coefficients on elements derived from $\mathbf{x}_i$, by contrast, are recovered only insofar as their group means differ. A characteristic unrelated to event timing, for example, would have approximately equal mean across trajectory groups and its coefficient would not be identified. The SIE, by contrast, works with observation-specific estimates throughout and thereby faces neither difficulty, pooling information across trajectory groups when estimating the effect-growth parameters while retaining the within-group variation in $\mathbf{x}^{(m)}_i$.
	
A similar consideration governs the aggregation weights. Since the regression returns cell parameters and their covariate slopes, it recovers an estimand $\sum_{it}w_{it}\tau^{(m)}_{it}$ exactly when the weights $w_{it}$ are constant within trajectory-period cells or linear within them in the interacted characteristics. The average event-specific effects $\tau^{(m)}_\ell$ satisfy the first of these, since all units in a trajectory group experience their $m$-th event in the same period. Estimands whose weights vary within cells in some other manner require the design to be extended accordingly, with the same consequences for dimensionality as above.
	
A notable advantage of the TWFE regression-based approach to estimating aggregate event-specific effects is that it provides an analytic variance for such quantities. Since every event-specific estimate is a linear combination of the total trajectory effects with weights fixed by $(\mathbf{d},\mathbf{x})$, the variance of any aggregate is the corresponding linear transformation of the variance-covariance matrix of the regression coefficients. Two considerations make this route unattractive. First, the transformation is cumbersome to derive beyond the first event since Stage~$m$.a subtracts the imputed path of every prior event. Second, such an analytic approach is available only where the aggregate event-specific effects of interest can be obtained as linear transformations of the coefficients of a single TWFE regression. As the preceding discussion explains, this holds only under conditions on the growth covariates $\mathbf{x}^{(m)}_i$ and the aggregation weights, and, in unbalanced panels, only asymptotically. While the SIE does not provide analytic standard errors, it remains consistent when these conditions do not hold and inference on the estimands it provides can be undertaken using a Bayesian bootstrap procedure outlined in the following section.
	
While it is therefore possible to obtain event-specific effect estimates from a TWFE regression, this paper proposes the imputation approach of the SIE as the primary estimator, owing to its greater flexibility in accommodating unit-level heterogeneity, arbitrary aggregation weights and unbalanced panels. 
	
\subsection{Inference}\label{sect:inference}
	
The Bayesian bootstrap \citep{rubin_bayesian_1981,lo_class_1987} provides inference on the SIE without recourse to an analytic variance. For a general estimand $\tau^{(m)}_w=\sum_{it\in\Omega^{(\geq m)}}w_{it}\tau^{(m)}_{it}$, the procedure is as follows.
	
\begin{enumerate}[leftmargin=2em,itemsep=0.3ex,topsep=0.4ex]
    \item Select the number of bootstrap draws, $B$.
    \item For each draw $b\in\{1,\ldots,B\}$:
    \begin{enumerate}[leftmargin=1.5em,itemsep=0.2ex,topsep=0.3ex]
        \item Draw unit-level weights $(\omega^b_1,\ldots,\omega^b_N)\sim\mathrm{Dirichlet}(1,\ldots,1)$, scaled to sum to $N$, applying $\omega^b_i$ to every observation of unit $i$ across time.\footnote{To accommodate dependence across units, such as that due to spatially-correlated shocks, the weights should instead be drawn at the level of the cluster and applied to every observation of every unit within it.}
        \item Re-estimate Stage~0 by weighted OLS on $\Omega^{(0)}$.
        \item Recompute $\widehat\lambda^{(m)b}_\ell$ by weighted OLS at each Stage~$m$.b, and all subsequent-stage estimates, using the same weights throughout.
        \item Aggregate the resulting $\widehat\tau^{(m)b}_{it}$ using the weights $w_{it}$ to obtain $\widehat\tau^{(m)b}_w$, recomputing any weights that depend on realised cell sizes from the draw's own effective sizes.
    \end{enumerate}
    \item Confidence intervals are obtained as percentiles of the empirical distribution of $\widehat\tau^{(m)b}_w$ across bootstrap draws.
\end{enumerate}

Validity requires that the estimator be a smooth functional of the unit-level empirical distribution. This holds in the present setting because each $\widehat\tau^{(m)}_{it}$ depends on unit $i$'s own observations and on the common nuisance parameters $\widehat\gamma_t$, $\widehat\pi_t$ and $\widehat\lambda^{(j)}_\ell$, which are finite in number when $T$ and $K$ are fixed. The aggregate is therefore an average across units of a function of each unit's data evaluated at a finite-dimensional first-step estimate, and the required smoothness follows under \ref{ass:reg} and the nonsingularity imposed by \ref{ass:supp}. The unit effects $\widehat c_i$ grow in number with the sample but are themselves functions of each unit's own untreated observations, and so do not disturb this structure.
	
The Dirichlet weights are preferred over those of the conventional nonparametric bootstrap because every unit retains strictly positive weight in every draw, so the comparison sets on which \ref{ass:supp} depends are never emptied and the Stage~$m$.b design matrix never becomes singular by chance. Such sets are by construction thin at long horizons and high event counts, so that under multinomial resampling some replications would fail or silently drop horizons, and the resulting interval would be conditioned on those draws that happened to succeed. \citet{callaway_difference_2021} make the same argument in motivating their multiplier bootstrap. That alternative is less convenient here, since perturbing an influence function presupposes that it is available in closed form, whereas re-executing the estimator under reweighting does not.
	
A final observation clarifies what object the resulting intervals cover. Proposition~\ref{prop:sie} establishes consistency for $\sum_{it}w_{it}\,\E[\tau^{(m)}_{it}\mid\mathbf{d},\mathbf{x}]$, which is a weighted average of conditional mean effects rather than of the effects realised among the particular units observed. The dispersion of $\tau^{(m)}_{it}$ about that mean is therefore part of the sampling error and enters the reported uncertainty alongside the variance of the disturbance $u_{it}$. \citet{borusyak_revisiting_2024} instead target a sum of realised unit-level effects, for which the dispersion in treatment effects is not part of the sampling error and must be separated from the disturbance. This is not possible without further assumptions, which is the source of the conservatism in their standard errors.
    
\subsection{Diagnostic Test of Effect-Growth Selection}\label{sect:test_eventrep}
	
As with the parallel trend assumption in single-event settings, \ref{ass:cpet} cannot be tested at the horizons at which it is invoked. For an observation in $\Omega^{(k)}$ at relative time $\ell$ after its $m$-th event, \ref{ass:cpet} is needed precisely to separate the effects of event-$m$ from those of later events. The estimator makes this separation exact, which leaves no unexplained variation against which the assumption could be assessed. Evidence on the selection it rules out can, however, be obtained from observations in $\Omega^{(m)}$, by comparing growth trajectories in event-$m$ effects among units that do and do not subsequently experience a further event.

Such a diagnostic test proceeds as follows. Fix an event $m$ and a relative time $\ell\geq1$, and consider the units entering the Stage~$m$.b estimation sample, all of which remain in $\Omega^{(m)}$ at that horizon. For $q\geq1$, let $g^{(m)}_{i\ell q}=\mathbf{1}[E^{(m+1)}_i=E^{(m)}_i+\ell+q]$ indicate that unit $i$ experiences its next event exactly $q$ periods later, and augment the Stage~$m$.b regression as
\begin{equation}\label{eq:cpet_test}
    \widehat\tau^{(m)}_{i,E^{(m)}_i+\ell}-\widehat\tau^{(m)}_{i,E^{(m)}_i}
    =\mathbf{x}^{(m)\prime}_i\lambda^{(m)}_\ell
    +\sum_{q=1}^{Q}\psi^{(m)}_{\ell q}g^{(m)}_{i\ell q}+\varsigma_{i\ell},
\end{equation}
where $Q$ is either the number of periods remaining in the panel after the earliest such $E^{(m)}_i+\ell$ or, if the cells at large $q$ are sparse, some lower terminal value. If the latter approach is pursued, $g^{(m)}_{i\ell Q}$ is redefined to indicate a next event at $Q$ or more periods later. Under \ref{ass:leg}, growth in the event-$m$ effect is mean independent of the full trajectory conditional on $\mathbf{x}^{(m)}_i$, so that $\psi^{(m)}_{\ell q}=0$ for every $q$. This is stronger than \ref{ass:cpet}, which restricts only the comparison between units that have and have not experienced a further event by horizon $\ell$. A non-zero coefficient indicates that units about to experience a further event were already on a different effect path. Plotting $\widehat\psi^{(m)}_{\ell q}$ against $q$ yields a figure directly analogous to the leads of a conventional event-study plot.

While the regression in equation \eqref{eq:cpet_test} pertains to a single event-horizon combination, estimands of interest will often rely on many $\left(m,\ell\right)$ pairs. The test statistic used to gauge the significance of departures from the test's null must therefore accommodate multiple testing. Formally, for an estimand $\tau^{(m)}_w$, let $\mathcal{L}^{(m)}_w$ collect the pairs $(j,\ell)$ such that $\widehat\lambda^{(j)}_\ell$ enters $\widehat\tau^{(m)}_{it}$ for some observation with $w_{it}\neq0$. The relevant null is that $\psi^{(j)}_{\ell q}=0$ at every element of $\mathcal{L}^{(m)}_w$ and every $q$. Since the elements of $\widehat{\boldsymbol\psi}$ are estimated on survivor samples that shrink as the horizon lengthens, and whose sizes remain in unequal proportion however large the panel, they should be studentised before comparison. An intuitive test statistic then reflects the largest departure from the null of $\psi^{(j)}_{\ell q}=0$ among the studentised elements of $\widehat{\boldsymbol\psi}$.
	
Implementation requires estimating \eqref{eq:cpet_test} alongside the Stage~$m$.b regression within every bootstrap replication, so that each draw $b$ delivers a full vector $\widehat{\boldsymbol\psi}^{b}$ under the same reweighting as the effect estimates. Writing $s^{(j)}_{\ell q}$ for the standard deviation of $\widehat\psi^{(j)b}_{\ell q}$ across draws, each bootstrap replication thus yields
\[
M^b=\max_{(j,\ell)\in\mathcal{L}^{(m)}_w,\;q\leq Q}
\bigl|\widehat\psi^{(j)b}_{\ell q}-\widehat\psi^{(j)}_{\ell q}\bigr|\big/s^{(j)}_{\ell q},
\]
and the null is rejected at level $\alpha$ if $\max|\widehat\psi^{(j)}_{\ell q}|/s^{(j)}_{\ell q}$ exceeds the $(1-\alpha)$ quantile of $M^b$ across draws. This quantile also delivers uniform confidence bands for the $\psi^{(j)}_{\ell q}$, so that the figure described above can be presented with bands whose exclusion of zero at any $q$ is equivalent to rejection of the joint null.
	
While informative, it should be acknowledged that this diagnostic test faces an important limitation. Because the test statistic is computed from the same estimated effects as the estimator itself, the diagnostic test cannot be treated as external to inference on $\widehat\tau^{(m)}_w$. By drawing on the same survivor samples, estimation errors in the lead coefficients $\widehat\psi^{(j)}_{\ell q}$ and the effect estimates $\widehat\tau^{(m)}_w$ are correlated, and restricting attention to samples in which the test has passed therefore truncates the distribution of the error in $\widehat\tau^{(m)}_w$. Intervals constructed from its unconditional distribution -- which is what the inference procedure outlined in Section~\ref{sect:inference} delivers -- then misstate coverage among the estimates actually reported. The direction of the distortion is undetermined in general and depends on the correlation between the two types of estimation error, which in turn reflects how the lead indicators covary with the growth covariates. \citet{roth_pretest_2022} shows that in the analogous case of pre-trend testing the resulting distortion can be substantial. The pre-trend test of \citet{borusyak_revisiting_2024} escapes this problem, since it is computed on untreated observations alone and hence is uncorrelated with their estimator. No analogue is available here, however, since the test's null concerns treatment effects, which are defined only for treated observations, so that no diagnostic can be constructed from data the estimator does not itself use. Conditional inference procedures of the kind developed for pre-trend testing could in principle be adapted to the present setting, and I leave this to future work.
	
	
Notwithstanding this limitation, rejection of the test in a given application provides indication that the data contradict \ref{ass:leg} and I therefore recommend withholding $\widehat\tau^{(m)}_w$ where the test rejects. Reporting the lead coefficients and their confidence intervals alongside the effect estimates is also recommended, to make clear which horizons the diagnostic can speak to and with what precision, and how close the realised statistic came to its critical value.
 
\subsection{Heterogeneity Analysis}\label{sect:heterogeneity}

In the same manner as the imputation estimator of BJS returns observation-specific total effects, the SIE returns estimates of event-specific effects for each event that has been experienced by a given observation. While individual estimates are uninformative in isolation, since each carries the outcome shock $u_{it}$ alongside the estimation error of the nuisance parameters, they nonetheless provide a granular basis for heterogeneity analysis.

Heterogeneity along discrete attributes can be examined via the choice of the weights $w_{it}$ used to define the particular ATT estimand of interest, so that they effectively condition the ATT on the characteristic of interest. If one wants to examine heterogeneity along continuous attributes or among a 
given attribute while conditioning on another, $\widehat\tau^{(m)}_{it}$ can be regressed on the covariates of interest to obtain a best linear predictor of event-specific effects.\footnote{The sample of such a regression will be dictated by the object of interest. To examine heterogeneity in the $h$-th lag effect of event $k$, for example, one would restrict the sample to $\left\{it \mid t-E_i^{(k)}=h\right\}$.} Since the regressand in this exercise is itself estimated, conventional standard errors from that regression will be invalid. The projection should instead be run within each bootstrap replication, with inference on its coefficients based on percentiles of the resulting distribution.

An alternative method to examine heterogeneity according to discrete attributes would be to estimate the SIE separately within attribute groups. This is equivalent to interacting the covariate-time slopes of the untreated model and every element of $\mathbf{x}^{(m)}_i$ with the group indicators 
and requires \ref{ass:supp} to hold within each group separately. Conditioning through the weights, which retains the pooled estimation of the growth relationship, is preferable unless the growth relationship itself is thought to differ by attribute group.

Because each unit's imputed path begins at its own realised impact effect, heterogeneity in the level of event effects is recovered along any dimension. Heterogeneity in their growth, however, is recovered only along dimensions spanned by $\mathbf{x}^{(m)}_i$, since for $it\in\Omega^{(k)}$ with $k>m$ event-$m$ effect at horizon $\ell$ is calculated exactly as the event-$m$ impact effect plus $\mathbf{x}^{(m)\prime}_i\widehat\lambda^{(m)}_{t-E^{(m)}_i}$. A characteristic that governs how event-$m$ effects evolve but is omitted from $\mathbf{x}^{(m)}_i$ will therefore appear unrelated to effect growth among the imputed observations, so that reported heterogeneity is attenuated in proportion to the share of observations for which effects are imputed.\footnote{Where the omitted determinant is unobserved and mean independent of $(\mathbf{d},\mathbf{x})$ given $\mathbf{x}^{(m)}_i$, \ref{ass:leg} continues to hold, the aggregate estimates remain consistent, and only the heterogeneity analysis is affected. Where it is an element of $\mathbf{x}_i$ excluded from $\mathbf{x}^{(m)}_i$, or is related to selection into subsequent events, \ref{ass:leg} fails and the SIE is in general biased.} Any attribute along which heterogeneity in effect growth is to be reported should therefore enter $\mathbf{x}^{(m)}_i$, and is accordingly restricted to functions of the time-invariant covariates, which \ref{ass:nbc} requires to be unaffected by treatment, and of the trajectory through event $m$. 
	
\section{Issues with Conventional TWFE Approaches}\label{sect:twfe}

As discussed in Section \ref{sect:wooldridge}, OLS estimation of a sufficiently flexible TWFE specification can recover consistent estimates of total treatment trajectory effects in settings with recurrent events. Section \ref{sect:estimator} further explains that if one imposes assumptions additional to those conventionally invoked in single-event settings, these estimates can be combined to recover estimates of occurrence-specific event effects. While there is thus nothing `wrong' with TWFE specifications \emph{per se} when it comes to estimating treatment effects in recurrent-event settings, the `sufficiently flexible' qualification is crucial. If treatment effects are heterogeneous, this condition requires that event indicators are replaced with interactions between period dummies and dummies that encode units' entire treatment trajectory (i.e. the timing of all events they experience), a specification that, to the author's knowledge, no paper undertaking event-study analysis in a recurrent-event setting has implemented. In light of this observation, this section evaluates several alternative TWFE specifications that have been implemented in existing literature, explaining when and why these alternatives fail to recover estimates that have clear interpretation as reasonably-weighted averages of event-specific effects.\footnote{The focus here is on intuitive explanation of these failures. Formal weighting representations, together with the conditions under which each specification is consistent, are provided in Appendix~\ref{app:twfe_proofs}.}
	
\paragraph{Single-Event Dynamic TWFE Specifications.} A common approach in applied work is to estimate a dynamic event study defining event indicators with respect to the first observed event occurrence only without modelling subsequent events. Such an approach is adopted in studies of natural disasters \citep{deryugina_fiscal_2017,patel_floods_2024}, in the job displacement literature following \citet{jacobson_earnings_1993} \citep[see also][]{couch_earnings_2010}, and in studies of health shocks \citep{dobkin_economic_2018}.
	
To fix ideas, consider the following TWFE model, which I refer to as the `single-event dynamic' TWFE specification (SE-TWFE).
\begin{equation}\label{eq:setwfe}
    Y_{it}=\sum_{h=0}^{\bar{h}^{(1)}}\tau_h\,w^{(1,h)}_{it}+\sum_{s=2}^{T}\tilde{\gamma}_s fs_t+\sum_{s=2}^{T}fs_t\,\mathbf{x}_i'\tilde{\pi}_s+\tilde{c}_i+\tilde{u}_{it},
\end{equation}
where $w^{(1,h)}_{it}\equiv\mathbf{1}[\ell^{(1)}_{it}=h]$ with $\ell^{(1)}_{it}=t-E^{(1)}_i$, and $\bar{h}^{(1)}=\max_i\,(T-E^{(1)}_i)$ is the maximum first-event horizon observed in the estimation panel. The parameters $\tilde{c}_i$, $\tilde{\gamma}_s$, $\tilde{\pi}_s$ and the disturbance $\tilde{u}_{it}$ are those of the specification rather than of the data generating process, and are distinguished accordingly with tildes. Relative to the DGP \eqref{eq:dgp_marginal}, untreated outcomes are correctly specified and the model is misspecified only through the omission of indicators for events after the first.
	
Practitioners adopting this specification, or similar single-event models, typically defend it in one of two ways. The first appeals to the exogeneity of event incidence. Where the occurrence of subsequent events is unrelated to unit-level unobservables, their omission is argued not to bias estimates of the $\tau_h$ coefficients, and supplementary analysis restricting the treated group to units experiencing a single event is offered as evidence that key results are unaffected \citep{deryugina_fiscal_2017}. The second concedes that the estimated $\tau_h$ reflect the effects of subsequent events alongside those of the first, but treats this combined object, which is essentially the total trajectory effect introduced above, as the target of estimation \citep{dobkin_economic_2018}.
	
These defences are less distinct than they appear. Exogeneity of incidence ensures that units observed at a given first-event horizon are not selected, but it does not sever the link between the omitted indicators and the retained ones, since a subsequent event can only occur once the first has taken place. Random incidence therefore delivers not the effect of the first event but a mixture of first- and subsequent-event effects, which is the object the second defence advocates. Both arguments are, in this sense, arguments for the total trajectory effect.
	
In light of these arguments, it is appropriate to judge the specification according to whether it is a consistent estimator for average total trajectory effects. As shown in Appendix~\ref{app:single_dynamic}, OLS applied to \eqref{eq:setwfe} returns coefficients $\widehat{\tau}_h$ that are weighted sums of the effects of every event active at each treated observation. Observations at first-event horizon $h$ carry weights summing to unity and those at every other horizon carry weights summing to zero. Since no subsequent-event indicators enter the specification, these weights depend on the configuration of first-event dates across the sample and on $\mathbf{x}_i$, but not on the timing of any subsequent event. In a balanced panel they are therefore constant within groups of observations sharing the same first-event date, horizon and covariate values. Consistency for the average total trajectory effect requires the expected sum of active event effects at a given first-event horizon, conditional on $\mathbf{x}_i$, to be invariant to both the date of the first event and to the covariates, which is essentially \ref{ass:cpt}, applied to trajectory effects rather than to untreated potential outcomes.
	
Two intuitive cases satisfy this condition. The first is where the sum of active effects is common across units at each first-event horizon, which is satisfied when both event-specific effects and inter-event gaps are homogeneous. The second is where inter-event gaps are common across units, so that the same set of events is active at each horizon, but event-specific effects vary, provided their conditional means are invariant to the timing trajectory and covariates. Homogeneity of event-specific effects alone does not suffice when the distribution of inter-event gaps varies with the first-event date or covariates, since the total effect at a given first-event horizon will still depend on the relative time elapsed since each subsequent event. 
	
When both event effects and inter-event gaps vary across units, the coefficients from \eqref{eq:setwfe} cannot be interpreted as a reasonable average of trajectory effects. Relative to the sufficiently flexible TWFE specification of Section~\ref{sect:wooldridge}, the failure arises because \eqref{eq:setwfe} fits a single coefficient $\tau_h$ across units whose trajectory effects differ. While an apparently innocuous simplification, it is consequential because of the FWL mechanics of OLS regression. The zero-sum weights on other horizons are harmless only when the trajectory effect at those horizons is common across the units receiving them. When it is not, the other-horizon weights, which need not be positive, deliver a combination of trajectory effects at other horizons whose contributions are governed by the joint distribution of first-event timing and inter-event gaps. A distinct difficulty is that when first-event timing differs, the set of units observed at horizon $h$ shifts with $h$, so that even in cases where each coefficient $\tau_h$ is a convex average of unit-specific effects, the sequence of estimates does not trace the effect trajectory of any fixed population. Conditioning on the timing configuration, the direct analogue of the cohort-based corrections of \citet{sun_estimating_2021} and \citet{callaway_difference_2021}, addresses both problems but cannot decompose the recovered trajectory into event-specific components, as explained in Section \ref{sect:wooldridge} and established below in \eqref{eq:collinearity}.
	
\paragraph{Pooled Dynamic TWFE Specifications.} Recognising that the SE-TWFE model will, at best, recover total trajectory effects, another commonly-applied TWFE specification attempts to isolate event-specific effects by pooling event-relative time indicators across event occurrences \citep{krolikowski_choosing_2018}. This approach is implemented using variants of the following specification, which I refer to as the `pooled dynamic' TWFE model (P-TWFE).
\begin{equation}\label{eq:ptwfe}
    Y_{it}=\sum_{h=0}^{\bar{h}^{(1)}}\tau_h\,\bar{w}^{(h)}_{it}+\sum_{s=2}^{T}\tilde{\gamma}_s fs_t+\sum_{s=2}^{T}fs_t\,\mathbf{x}_i'\tilde{\pi}_s+\tilde{c}_i+\tilde{u}_{it},
\end{equation}
where $\bar{w}^{(h)}_{it}\equiv\sum_{k=1}^{K}w^{(k,h)}_{it}$ replaces the single-event indicators \eqref{eq:setwfe}, with horizon-$h$ indicators that are pooled across event occurrences. Since simultaneous events are ruled out, at most one occurrence can be dated $h$ periods before any given observation, and hence $\bar{w}^{(h)}_{it}$ is a binary indicator for some event having occurred exactly $h$ periods earlier. This specification is algebraically identical to distributed lag models that have been used to estimate dynamic event effects while controlling for repeat exposures \citep{strobl_economic_2011,hsiang_environment_2014}, since the $h$-th lag of $\sum_{k}\mathbf{1}[t=E^{(k)}_i]$ is precisely $\bar{w}^{(h)}_{it}$. As with \eqref{eq:setwfe}, untreated outcomes are correctly specified and the model is potentially misspecified only in terms of the event structure.
	
Relative to SE-TWFE, P-TWFE is closer to the DGP of equation \eqref{eq:dgp_marginal} in that it encodes every event occurrence rather than the first alone. Since the coefficients $\tau_h$ carry neither unit nor occurrence notation, however, the specification is correctly specified only if dynamic event effects are common across units and occurrences. Proposition~\ref{prop:pooled} in Appendix~\ref{app:pooled_dynamic} shows it is consistent under the weaker condition that conditional mean event effects are common across occurrences and invariant to $(\mathbf{d},\mathbf{x})$. The same proposition shows that OLS applied to \eqref{eq:ptwfe} returns coefficients $\widehat{\tau}_h$ that are weighted sums of the effects of every event active at each treated observation, taken across occurrences and horizons alike. When either homogeneity across occurrences or mean independence fails, however, these weighted sums lack clear interpretation as average event effects. 
	
When distinct event occurrences have heterogeneous effects, the $\widehat{\tau}_h$ of the P-TWFE model will reflect a combination of $\tau^{(1)}_{h},\tau^{(2)}_{h},\ldots$, with the precise composition governed by the FWL weights. Since the shares of observations at horizon $h$ arising from first, second and later events are determined by the joint distribution of event timings, the implied weighting varies with $h$ and the estimated profile $\{\widehat{\tau}_h\}_{h\geq0}$ traces the dynamics of no single occurrence. The consequences of this are most serious where occurrence heterogeneity follows a systematic pattern. Because a $k$-th event leaves fewer remaining periods of observation than the first, the share of observations at horizon $h$ arising from later occurrences falls as $h$ rises, so that the short-horizon coefficients draw more heavily on later occurrences than the long-horizon coefficients do. Suppose that units adapt to repeated exposure, so that the effect of each successive event is attenuated relative to the last at every horizon. The short-horizon coefficients then draw disproportionately on the attenuated later occurrences and the long-horizon coefficients on the first, and the estimated profile will appear to grow in magnitude with the horizon even where the profile of each occurrence is flat. Were subsequent events instead more damaging than the first, the same mechanism would generate spurious attenuation. The estimated profile therefore confounds the dynamics of event effects with the changing composition of occurrences across horizons, so that the shape of the response, which is the feature a distributed lag specification is typically estimated to recover, cannot be read from the $\widehat{\tau}_h$ coefficients.
	
A second difficulty arises from heterogeneity across units. Where event effects vary across units in a manner correlated with event timing or covariates, each $\widehat{\tau}_h$ is a weighted sum of unit-specific event effects with weights that need not be positive and that are determined by the joint distribution of event timings rather than by any feature of the estimand. The practical implication is that coefficient estimates from pooled and distributed lag specifications lack clean interpretation as average effects, either of a particular event occurrence or of all occurrences taken together. Recovering occurrence-specific effects requires occurrence-specific parameters, which is the modification considered next.
	
\paragraph{Multi-Event Dynamic TWFE Specifications.} The most saturated of the specifications considered here includes separate event-relative time indicators for every occurrence:
\begin{equation}\label{eq:metwfe}
    Y_{it}=\sum_{k=1}^{K}\sum_{h=0}^{\bar{h}^{(k)}}\tau^{(k)}_h\,w^{(k,h)}_{it}+\sum_{s=2}^{T}\tilde{\gamma}_s fs_t+\sum_{s=2}^{T}fs_t\,\mathbf{x}_i'\tilde{\pi}_s+\tilde{c}_i+\tilde{u}_{it},
\end{equation}
where $w^{(k,h)}_{it}\equiv\mathbf{1}[\ell^{(k)}_{it}=h]$ and $\bar{h}^{(k)}=\max_i\,(T-E^{(k)}_i)$ is the maximum horizon observed after a $k$-th event. As with \eqref{eq:setwfe} and \eqref{eq:ptwfe}, pre-event observations serve as the reference category, the covariate--time interactions of \eqref{eq:dgp_marginal} are retained, and the parameters carrying tildes are those of the specification rather than of the data generating process. In what follows, I refer to \eqref{eq:metwfe} as the `multi-event dynamic' TWFE model (ME-TWFE).
	
\citet{sandler_multiple_2014} recommend this specification on the basis of Monte Carlo evidence that it performs relatively well when events reoccur. Despite this recommendation, no paper appears to have implemented it in an applied recurrent-event setting. It is nonetheless more robust than either the SE-TWFE or P-TWFE models, since it imposes neither the omission of subsequent events that characterises \eqref{eq:setwfe} nor the pooling across occurrences that characterises \eqref{eq:ptwfe}. 
	
While ME-TWFE accommodates heterogeneity in effects across event occurrences, it requires that conditional mean effects be invariant to the timing trajectory and to the covariates and fails to recover consistent estimates otherwise. To see this note that, relative to the sufficiently flexible specification of Section~\ref{sect:wooldridge}, \eqref{eq:metwfe} omits the interactions of the event indicators with the timing trajectory. It therefore fits a single coefficient $\tau^{(k)}_h$ across units whose $k$-th event effects at horizon $h$ may differ. As shown in Appendix~\ref{app:multi_dynamic}, $\widehat{\tau}^{(k)ME-TWFE}_h$ is thus a weighted sum of these unit-specific effects whose weights are determined by the configuration of event timings across the sample and may be negative. 
	
	
Invariance of conditional mean effects to the timing trajectory is the multi-event analogue of the restriction that the cohort-based corrections of \citet{callaway_difference_2021} and \citet{sun_estimating_2021} avoid imposing. Those papers show that contamination in single-event TWFE arises from aggregation across treatment cohorts, and resolve it by estimating cohort-specific effects from clean comparisons before aggregating with researcher-chosen weights. The natural extension here would interact each $w^{(k,h)}_{it}$ with dummies for the timing trajectory $E_i$, the multi-event analogue of the cohort, and aggregate the resulting trajectory-specific estimates. This is precisely the specification that Section~\ref{sect:wooldridge} shows to be sufficiently flexible for total trajectory effects.
	
Applied to event-specific effects, however, the extension fails, because the interacted regressors are exactly linearly dependent. Consider a trajectory $e$ in which a unit experiences its first two events at $e^{(1)}$ and $e^{(2)}$, with gap $G=e^{(2)}-e^{(1)}$. For every $h\geq G$,
\begin{equation}\label{eq:collinearity}
    w^{(1,h)}_{it}\cdot\mathbf{1}[E_i=e]=w^{(2,h-G)}_{it}\cdot\mathbf{1}[E_i=e],
\end{equation}
since both indicators select the single observation at $t=e^{(1)}+h$. Within any trajectory group experiencing two or more events, the interacted first- and second-event indicators are therefore perfectly collinear at every horizon at which both events are active, and the same holds for every later pair of occurrences. Interacting the event indicators with the trajectory dummies thus precludes separate identification of event-specific effects because the two event-specific regressors are perfectly collinear within each treatment trajectory cell. This is the sense in which the decomposition problem of Section~\ref{sect:wooldridge} reappears as a property of the design matrix. Conditioning on the trajectory can therefore recover the total effect of that trajectory but not the event-specific effects of which it is composed. Recovering the latter requires borrowing information across trajectory groups, using observations at which only an earlier event is active to identify its continuing contribution where a later event has intervened, which is the role played by \ref{ass:cpet} in Section~\ref{sect:estimator}.
	
The discussion of this section, underpinned by the formal Propositions~\ref{prop:single_dynamic}, \ref{prop:pooled}  and \ref{prop:multi_dynamic} in Appendix~\ref{app:twfe_proofs}, demonstrates that TWFE specifications implemented in recurrent-event settings recover interpretable estimates of event-specific effects only under conditions considerably stronger than those of Section~\ref{sect:wooldridge}. The single-event specification is not designed to recover event-specific effects and delivers total trajectory effects only where the conditional mean trajectory effect at each horizon is invariant to first-event timing. The pooled specification requires effects be common across occurrences and invariant to event timing. The multi-event specification admits occurrence-specific effects but maintains the requirement that their conditional means be invariant to the timing trajectory. Imposing the flexibility that Section~\ref{sect:wooldridge} shows to be sufficient for total trajectory effects is of no help, since by \eqref{eq:collinearity} it renders event-specific effects unidentified. These results motivate the sequential imputation estimator of Section~\ref{sect:estimator}.
	
\section{Monte Carlo Simulations}\label{sect:montecarlo}
	
This section presents Monte Carlo simulations that compare performance of the SIE against that of the TWFE specifications described in Section \ref{sect:twfe}.
	
\subsection{Data Simulation}\label{sect:dgps}
	
The Monte Carlo exercise simulates four distinct data generating processes (DGPs)  \mcnsim\ times. Data in each run is simulated for a panel of 1000 units over 30 periods according to $$Y_{it} = \alpha_i + \beta_t + \sum_{k=1}^{K} \sum_{h \geq 0} \tau_{i,h}^{(k)} \cdot \mathbf{1}[\ell_{it}^{(k)} = h] + \varepsilon_{it},$$
where $\varepsilon_{it} \sim N(0, 0.25)$. Fixed effects are i.i.d.\ standard normal random variables and time effects are simulated as $\beta_t = 0.1t + \mu_t$, where $\mu_t \sim N(0, 0.09)$. No time-invariant covariates enter the untreated outcome, so Assumption \ref{ass:lin} holds with unit and period effects alone.
	
In each of the four DGPS, units can experience up to two treatment events (i.e. $K=2$). 90\% of units experience a first event, with the period of first-event occurrence simulated as a discrete random variable with uniform probability on the support $E_i^{(1)} \in \{3, \ldots, \bar{E}^{1}\}$, where $\bar{E}^{1} = 25$ unless stated otherwise. The earliest period in which an event can occur is therefore the third, which ensures at least two periods of pre-event data per unit with which unit fixed effects can be estimated.\footnote{While the period of first treatment is rarely observed directly in empirical settings, Section~\ref{sect:implementation} discusses how such unobserved initial conditions can be addressed in practice.} A share $p_2$ of first-event units also experience a second event, with second-event timing determined by the gap $G_i = E_i^{(2)} - E_i^{(1)}$, which is simulated as a discrete random variable with uniform probability on the support $G_i \in \{2, \ldots, \min(\bar{G},\, 30 - E_i^{(1)} - 2)\}$. Baseline values are $p_2 = 0.5$ and $\bar{G} = 25$.
	
Treatment effects in all DGPs take the form
$$\tau_{i,\ell}^{(k)} = a_i^{(k)} + \delta_i^{(k)}\left(1 - e^{-\rho_i \ell}\right).$$ $a_i^{(k)}$ captures the unit-specific effect of event $k$ on impact, while $\delta_i^{(k)}$ and $\rho_i$ govern the extent and speed of subsequent recovery towards $a_i^{(k)} + \delta_i^{(k)}$. The growth increment restricted by Assumption \ref{ass:leg} is therefore
$$\Delta_{i\ell}^{(k)} \equiv \tau_{i,\ell}^{(k)} - \tau_{i,0}^{(k)} = \delta_i^{(k)}\left(1 - e^{-\rho_i \ell}\right),$$
which is a function of $\delta_i^{(k)}$ and $\rho_i$ but not of $a_i^{(k)}$. 
	
Baseline effect magnitudes attenuate proportionally with event occurrence,
$$a_i^{(1)} \sim N(-2, 0.25), \quad a_i^{(2)} \sim N(-1, 0.25),$$
$$\delta_i^{(k)} = I_i^{(k)} + \nu_i^{(k)}, \quad I_i^{(1)} \sim N(1.5, 0.16), \quad I_i^{(2)} \sim N(0.75, 0.16), \quad \nu_i^{(k)} \sim N(0, 0.09),$$
so that the first event has an average effect of $-2$ on impact, recovering gradually at rate $\rho_i$ towards a long-run average of $-0.5$, whereas the second event's impact and recovery are both half as large on average. The extent of recovery $\delta_i^{(k)}$ has an observable `intensity' component $I_i^{(k)}$, which enters the growth covariate vector, and an unobservable component $\nu_i^{(k)}$, which does not.
	
The \ref{ass:cpet} assumption restricts the conditional mean of growth in treatment effects given the growth covariate vector $\mathbf{x}_i^{(m)}$, with \ref{ass:leg} requiring it to be linear in that vector. Whether they hold therefore depends on the specification a researcher adopts as well as on the process generating the data. The baseline specification is $\mathbf{x}_i^{(m)} = (1, I_i^{(m)})'$, containing a constant and the event-$m$ intensity. Since $\delta_i^{(k)}$ is the sum of $I_i^{(k)}$ and an independently-drawn unobservable, the conditional mean of the growth increment is linear in this specification when the rate of recovery $\rho_i$ is common across units and when $\delta_i^{(k)}$ is itself linear in intensity. Where a DGP departs from either of these conditions, the baseline no longer satisfies \ref{ass:leg}, and I additionally report an enriched specification which restores it.
	
The four DGPs I consider differ in how they simulate treatment effects, second-event timing and selection into subsequent events. 
	
\paragraph{DGP 1. Independent Heterogeneous Effects.}
The first DGP features common growth in event effects across units, $\rho_i = 0.4$, and selection into the second event independent of all other features of the DGP. Specifically, the probability that each unit who experiences one event experiences a second event is constant at $p_2 = 0.5$. Under the baseline specification outlined above, the conditional mean of the growth increment is $E[\Delta_{i\ell}^{(m)} \mid d, x] = I_i^{(m)}(1 - e^{-0.4\ell})$, which is exactly linear in $x_i^{(m)}$ at every horizon. \ref{ass:leg} therefore holds and the SIE should recover unbiased estimates. 
	
\paragraph{DGP 2. Selection on Event Levels.} 
Treatment effects in this DGP are simulated in the same manner as in DGP 1. The selection process into the second event, however, differs. Rather than each first-event unit having an equal 50\% probability of experiencing a second event, this probability depends on the first-event impact effect
$$P\!\left(\text{second event} \mid a_i^{(1)}\right) = \Phi\!\left(-1.6 \cdot \frac{a_i^{(1)} + 2}{0.5}\right).$$
The parameters are calibrated so that $E[P(\text{second event})] = 0.5$ (at the mean $a_i^{(1)} = -2$, $\Phi(0) = 0.5$), but units with larger first-event impact effects are considerably more likely to experience a second event. Growth in event effects, however, remains independent of selection into the second event so that \ref{ass:leg} continues to hold.
	
\paragraph{DGP 3. Omitted Growth Driver.}
The third DGP introduces heterogeneity into the rate at which event effects recover, determined by a unit-specific type:
$$\rho_i = \begin{cases} 0.15 & \text{if } Z_i = 1 \quad (\text{persistent type}),\\ 0.70 & \text{if } Z_i = 0 \quad (\text{transient type}),\end{cases} \qquad Z_i \sim \text{Bernoulli}(0.5),$$
with impact effects and intensities as in DGP 1. Selection into the second event depends on the same type,
$$P\!\left(\text{second event} \mid Z_i\right) = p_2 + \kappa_Z\left(Z_i - \tfrac{1}{2}\right).$$
I set $\kappa_Z = 0.6$, which implies persistent types experience a second event with probability $0.8$ and transient types with probability $0.2$. 
	
Since $Z_i$ determines effect growth, its omission from the baseline specification of the growth regression violates both \ref{ass:cpet} and \ref{ass:leg}. As $\kappa_Z$ increases, the distribution of $Z_i$ among the observations in $\Omega^{(1)}$ used to estimate the growth parameters $\lambda_\ell^{(1)}$ becomes progressively less representative of the distribution among the observations in $\Omega^{(2)}$ for which the growth regression is used to compute event-1 effects. The imputation, and hence the SIE, will consequently become progressively biased as $\kappa_Z$ increases.
	
To redress this bias, I also consider an enriched specification of $x_i^{(m)} = (1, I_i^{(m)}, Z_i, Z_i I_i^{(m)})'$, under which\footnote{The first equality holds because $\nu^{(m)}_i$ is independent of the event timings and of $x_i^{(m)}$, while $\rho_i$ is determined by $Z_i$. The second holds because $Z_i$ is binary, so that $e^{-\rho(Z_i)\ell}=Z_ie^{-\rho(1)\ell}+(1-Z_i)e^{-\rho(0)\ell}$, where $c_z(\ell)\equiv1-e^{-\rho(z)\ell}$.}
$$E\!\left[\Delta_{i\ell}^{(m)} \mid d, x\right] = I_i^{(m)}\left(1 - e^{-\rho(Z_i)\ell}\right) = I_i^{(m)} c_1(\ell) Z_i + I_i^{(m)} c_0(\ell)\left(1 - Z_i\right).$$
This is exactly linear in $\mathbf{x}_i^{(1)}$ and hence restores \ref{ass:leg} (and, by implication, \ref{ass:cpet}). Together with the baseline specification, this DGP thus examines the extent to which misspecification of the growth covariate vector impairs the performance of the SIE.
	
\paragraph{DGP 4. Heterogeneous Effect Growth.}
Whereas DGP 3 explores performance of the SIE when \ref{ass:cpet} is violated by the omission of a determinant of growth from $\mathbf{x}_i^{(m)}$, this DGP explores performance when every growth determinant is included in $\mathbf{x}_i^{(m)}$, so that \ref{ass:cpet} holds, but \ref{ass:leg} is violated through functional form. The extent of recovery is exponential rather than linear in intensity:
$$\delta_i^{(1)} = 1.5\, e^{I_i - 1.5}+ \nu_i^{(1)}, \qquad \delta_i^{(2)} = 0.75\, e^{I_i - 1.5}+ \nu_i^{(2)}, \qquad \rho_i = 0.4, \quad \nu_i^{(k)} \sim N(0, 0.09).$$
Intensity $I_i$ is common across events, with $I_i \sim N(1.5, 0.16)$ truncated to $[0.3, 2.7]$, and is correlated with selection into the second event according to
$$P\!\left(\text{second event} \mid I_i\right) = \Phi\!\left(2.5 \cdot \frac{I_i - 1.5}{0.4}\right).$$
Compared to DGP 3, where the omission of $Z_i$ from the growth covariate vector compromises both \ref{ass:cpet} and \ref{ass:leg}, this DGP separates the two conditions. Conditional on $I_i$, growth in event effects is a deterministic function of $I_i$ plus a disturbance drawn independently of selection, so \ref{ass:cpet} holds. \ref{ass:leg} is violated, however, since the deterministic function is exponential rather than linear. An enriched specification augments the baseline of $\mathbf{x}_i^{(m)} = (1, I_i^{(m)})'$ with $(I_i^{(m)})^2$, to examine how far a polynomial approximation can accommodate a non-linear growth relationship.

\subsection{Performance of the Sequential Imputation Estimator}\label{sect:mc_sie}
	
For each Monte Carlo replication I estimate the SIE with \mcnbreps\ Bayesian bootstrap draws, constructed as described in Section~\ref{sect:inference}. I evaluate performance of the SIE by computing the average bias across Monte Carlo runs against the population ATT implied by the DGP, and by calculating `coverage' as the percentage of runs in which the 95\% percentile interval covers the DGP's true population ATT.
	
Table~\ref{tab:mc_sie} reports the results. Under DGP 1 the SIE is unbiased at both events and coverage of the bootstrap confidence intervals is nominal. The diagnostic test also exhibits nominal size, rejecting in approximately 5\% of runs. Under DGP 2 the results are essentially the same as those from DGP 1, despite the strong dependence of second-event occurrence on the first-event impact effect. This demonstrates how the SIE is consistent when selection into subsequent events is correlated with the magnitude of prior-event effects, which is similar to how conventional single-event difference-in-differences can accommodate differences in levels of untreated outcomes between treated and control groups. 
	
\begin{table}[H]
    \centering
    \caption{Performance of the Sequential Imputation Estimator}\label{tab:mc_sie}
\begin{tabular}{clcrrrr}
\hline\hline
\rule{0pt}{2.6ex}DGP & $x^{(1)}$ & Event & True ATT & Bias & Coverage & Rejection \\
\hline
1 & $1, I$ & 1 & -0.9924 & -0.0001 & 0.948 & 0.055 \\
 &  & 2 & -0.4947 & -0.0002 & 0.953 & 0.045 \\
\hline
2 & $1, I$ & 1 & -0.9927 & +0.0003 & 0.939 & 0.046 \\
 &  & 2 & -0.4956 & -0.0015 & 0.940 & 0.048 \\
\hline
3 & $1, I$ & 1 & -1.1079 & +0.0348 & 0.770 & 1.000 \\
 &  & 2 & -0.6399 & -0.1362 & 0.372 & 1.000 \\
 & $1, I, Z, Z \times I$ & 1 & -1.1079 & +0.0004 & 0.945 & 0.047 \\
 &  & 2 & -0.6430 & +0.0026 & 0.938 & 0.046 \\
\hline
4 & $1, I$ & 1 & -0.9095 & -0.0021 & 0.946 & 0.071 \\
 &  & 2 & -0.2967 & +0.0724 & 0.807 & 0.156 \\
 & $1, I, I^{2}$ & 1 & -0.9095 & +0.0020 & 0.946 & 0.066 \\
 &  & 2 & -0.2967 & -0.0054 & 0.949 & 0.053 \\
\hline\hline
\end{tabular}

    \begin{minipage}{0.95\textwidth}
        \vspace{0.25cm}
        \footnotesize Note: table summarises results from \mcnsim\ replications of each data generating process, with the SIE estimated using \mcnbreps\ Bayesian bootstrap draws. True ATT is the population average effect on the treated implied by the DGP, averaged over horizons $h = 0, \ldots, \mcmaxh$. Bias is the mean difference
        between the estimate and that quantity, averaged over the same horizons, and is reported in levels. Coverage is the mean rate at which the 95\% percentile interval contains the true ATT. Rejection is the rate at which the diagnostic test of Section~\ref{sect:test_eventrep} rejects at a nominal size of 0.05. 
    \end{minipage}
\end{table}
	
DGP 3, by contrast violates \ref{ass:cpet}, and by consequence \ref{ass:leg}, by making selection into the second event correlated with \emph{growth} in first-event effects. Under the baseline specification of $\mathbf{x}_i^{(1)}=(1,I^{(1)}_i)$ the SIE carries bias of \mcsieBiasThreeBaseTwo\ at the second event, about 21\% of the true second-event ATT, and coverage falls to \mcsieCovThreeBaseTwo.\footnote{Under DGP 3's enriched specification, \mcrefusalRateThreeSat\ of replications are refused at second-event horizons $h \geq 2$, since that specification alone contains four regressors and the survivor sample occasionally falls below the rank requirement of \ref{ass:supp}. Figures reported in text and those in the relevant row of Table \ref{tab:mc_sie} are computed on the remaining replications, with the true ATT averaged over the same horizons so that it remains comparable to the bias.}  The diagnostic test has strong power, rejecting in all of the 1000 runs. Consistency of the SIE is restored by extending the growth covariate vector to additionally include the selection determinant $Z_i$, along with its interaction with first-event heterogeneity.
	
DGP 4 demonstrates performance of the SIE when \ref{ass:cpet} holds but when \ref{ass:leg} does not. In this DGP, selection into the second event is uncorrelated with growth in first-event effects but the SIE's effect-growth regression is misspecified. The functional form failure causes modest bias in the first-event ATT estimate, with the second-event ATT exhibiting larger bias. The extent of bias in the first-event effects is so modest that coverage of the first-event confidence intervals remains close to nominal, with the diagnostic test rejecting in only 7\% of cases. Coverage of the second-event ATT falls to around 80\% with the diagnostic test rejection rate rising to 16\%. Enriching the growth covariate vector with $(I_i^{(1)})^2$ to accommodate greater non-linearity in the growth relationship reduces the second-event bias to $-0.005$ and restores coverage to around 95\%, while first-event bias is negligible under both specifications.

The contrast in rejection rates across DGPs 3 and 4 helps clarify what the diagnostic test measures. By augmenting the growth regression with indicators for the timing of the next event, equation~\eqref{eq:cpet_test} indicates whether growth in prior-event effects varies with the timing of subsequent events conditional on $\mathbf{x}_i^{(m)}$. Its null is an implication of \ref{ass:leg}, and the test has power against departures from \ref{ass:leg} that covary with the timing of subsequent events. This includes the selection on effect growth that \ref{ass:cpet} rules out, but not functional-form error that is unrelated to next-event timing. When failure of \ref{ass:leg} is due to failure of \ref{ass:cpet}, DGP 3 indicates the test has strong power. The results from DGP 4, by contrast, indicate that the diagnostic test has weak power against failures of \ref{ass:leg} that leave \ref{ass:cpet} intact. In DGP 4 every systematic determinant of effect growth is included in $\mathbf{x}_i^{(1)}$, and misspecification arises from the non-linear dependence of growth on first-event intensity. The test retains some power only indirectly, since first-event intensity drives both the misspecification error and selection into the second event. A pass on the diagnostic test is therefore best treated as evidence that \ref{ass:cpet} holds and that any remaining misspecification is not strongly correlated with selection into subsequent events. Where \ref{ass:cpet} holds given the included covariates, misspecifying how they enter the growth model biases the estimator only to the extent that the distribution of $\mathbf{x}_i^{(m)}$ differs between the units on which growth is estimated and those for whom effects are imputed. Choices of $\mathbf{x}_i^{(m)}$ should therefore be assessed both with the diagnostic test and by comparing the distribution of $\mathbf{x}_i^{(m)}$ across the two populations.
	
In summary, the Monte Carlo simulations demonstrate the SIE recovers approximately unbiased estimates of event-specific ATTs in a range of DGPs, with confidence intervals calculated using the Bayesian bootstrap exhibiting nominal coverage when \ref{ass:sutva}, \ref{ass:anticipation}, \ref{ass:nbc}, \ref{ass:lin}, \ref{ass:supp}, \ref{ass:leg} and \ref{ass:reg} hold. The diagnostic test of these assumptions proposed in Section \ref{sect:test_eventrep} has strong power against violations of \ref{ass:cpet}, which is necessary but not sufficient for \ref{ass:leg}, but only weak power against violations of \ref{ass:leg} that leave \ref{ass:cpet} intact.
	
\subsection{Comparison with Two-Way Fixed Effects Specifications}\label{sect:mc_twfe}
	
Table~\ref{tab:mc_twfe} reports bias for the three TWFE specifications of Section~\ref{sect:twfe} alongside that of the SIE.\footnote{Each estimator estimates a separate coefficient at every observed relative time, so that no restriction is imposed on effects at long horizons. This circumvents bias due to binning or truncating events that applied work routinely employs.} Considering first the ME-TWFE specification, Table~\ref{tab:mc_twfe} shows its performance across the four DGPs is uneven. Under DGP 1, both ME-TWFE and the SIE are unbiased at both events. Effects are heterogeneous across units in this DGP, but selection into the second event is independent of them, so the subpopulation from which second-event coefficients are identified has the same distribution of first-event effects as the population. The single coefficient $\tau^{(1)}_h$ that ME-TWFE fits across units is therefore the correct average for the doubly-treated units as well as for the rest, and nothing remains to contaminate $\widehat\tau^{(2)}_h$.
	
\begin{table}[H]
    \centering
    \caption{Bias of SIE and Two-Way Fixed Effects Specifications}\label{tab:mc_twfe}
%
\begin{tabular}{ccrrrrr}
\hline\hline
\rule{0pt}{2.6ex}DGP & Event & \multicolumn{2}{c}{SIE} & ME-TWFE & SE-TWFE & P-TWFE \\
\cline{3-4}
 & & Baseline & Extended & & & \\
\hline
1 & 1 & -0.0001 & --- & -0.0000 & -0.0926 & +0.2432 \\
 & 2 & -0.0002 & --- & -0.0006 & --- & --- \\
\hline
2 & 1 & +0.0003 & --- & +0.0426 & -0.0927 & +0.1329 \\
 & 2 & -0.0015 & --- & -0.3074 & --- & --- \\
\hline
3 & 1 & +0.0348 & +0.0004 & +0.0114 & -0.1082 & +0.2124 \\
 & 2 & -0.1362 & +0.0026 & -0.0454 & --- & --- \\
\hline
4 & 1 & -0.0021 & +0.0020 & -0.0656 & -0.0723 & +0.4708 \\
 & 2 & +0.0724 & -0.0054 & +0.5028 & --- & --- \\
\hline\hline
\end{tabular}

    \begin{minipage}{0.95\textwidth}
        \vspace{0.25cm}
        \footnotesize Note: mean bias over horizons $h = 0, \ldots, \mcmaxh$, calculated over \mcnsim\ replications. Each TWFE specification is estimated with a separate coefficient at every observed relative time. SE-TWFE contains first-event indicators alone and so yields no second-event estimate. P-TWFE pools coefficients across event occurrences and is reported against the first event only. Results in the `Extended' SIE column correspond to the rows of Table~\ref{tab:mc_sie} that feature the larger growth covariate vector.
    \end{minipage}
\end{table}
	
Under DGP 2 this ceases to hold, and ME-TWFE carries bias of \mcmeBiasTwoTwo\ at the second event against the SIE's \mcsieBiasTwoTwo. Since units with larger first-event impact effects are considerably more likely to experience a second event, the doubly-treated units have systematically more negative first-event effects than the pooled average. The first-event indicators absorb only that pooled average, and the residual first-event effect is then attributed to the event-2 coefficients. Selection into the second event in DGP 3 is based on type rather than on the level of first-event effects, and ME-TWFE's bias is correspondingly smaller, at \mcmeBiasThreeBaseTwo. The type $Z_i$ in this DGP governs the rate at which effects recover rather than their magnitude on impact, so that doubly-treated units differ from the pooled average in the shape of their first-event path rather than in its level. The residual left unabsorbed by the first-event indicators is therefore smaller, and the contamination of the second-event coefficients milder than in DGP 2. Bias in the ME-TWFE estimates is lower here than the SIE's estimates, by a factor of \mcmeAdvantageThree, under the specification that omits $Z_i$. DGP 4 produces the largest bias in the table of \mcmeBiasFourBaseTwo\ at the second event, which is sufficient to reverse the sign of the estimated effect. Selection in this DGP operates on intensity and the extent of recovery is exponential in it, so the doubly-treated units' first-event effects exceed the pooled average by a considerable margin. The residual attributed to the second event is correspondingly large. 
	
These results highlight that comparison between the SIE and ME-TWFE rests on the process governing selection into subsequent events. ME-TWFE is unbiased when such selection is independent of prior-event effects, and biased when it is not, with the magnitude of the bias governed by how far the doubly-treated units' prior-event effects depart from the population average. The SIE is unbiased in both cases, since it nets out each unit's own prior-event effects rather than a coefficient common across units, and fails only where its growth model is misspecified.

Turning now to the SE-TWFE and P-TWFE results, Table \ref{tab:mc_twfe} shows that the single-event and pooled specifications are biased against first-event effects in every DGP, with absolute bias of 0.07--0.11 (SE-TWFE) and 0.13--0.47 (P-TWFE), compared with at most 0.035 for the SIE under the baseline growth specification and at most 0.002 under the extended one. This bias persists despite the features that distinguish the DGPs from one another, since the sources of their bias are the omission of subsequent events and the pooling of effects across event occurrences respectively, rendering both misspecified across DGPs. As explained in Section \ref{sect:twfe}, the SE-TWFE model is best viewed as an estimator of average total trajectory effects, whereas the P-TWFE model is presumably intended to capture an event-specific trajectory averaged over event occurrences. Judging the SE-TWFE and P-TWFE estimates by how accurately they recover first-event effects is therefore arguably an unreasonable test. 

To further explore performance of the SE-TWFE estimator, Table~\ref{tab:mc_setwfe} compares the SE-TWFE estimates against the average total trajectory effect realised among the units in each replication. Against this quantity the estimator performs an order of magnitude better than against the event-specific effects of Table~\ref{tab:mc_twfe}, with bias below three per cent of the effect at every DGP and horizon group. Since the SE-TWFE specification includes first-event indicators alone, its coefficient at horizon $h$ absorbs whatever the outcome reflects at that point in a unit's trajectory, including the effects of any subsequent event the unit has by then experienced. When measured against the first-event effect this constitutes misspecification whereas when measured against the total trajectory effect it is close to the quantity the specification estimates. The residual bias shown in Table~\ref{tab:mc_setwfe} is due to the specification fitting a single coefficient across units whose trajectories differ.

\begin{table}[H]
    \centering
    \caption{Bias of the SE-TWFE Specification Against the Average Total Trajectory Effect}\label{tab:mc_setwfe}
%
\begin{tabular}{clrr}
\hline\hline
\rule{0pt}{2.6ex}DGP & $h$ & Trajectory ATT & SE-TWFE bias \\
\hline
1 & 0 & -2.0002 & -0.0102 \\
 & 1--4 & -1.1831 & -0.0050 \\
 & 5--8 & -0.7714 & -0.0012 \\
\hline
2 & 0 & -2.0008 & -0.0108 \\
 & 1--4 & -1.1836 & -0.0049 \\
 & 5--8 & -0.7723 & -0.0008 \\
\hline
3 & 0 & -1.9995 & -0.0280 \\
 & 1--4 & -1.2751 & -0.0254 \\
 & 5--8 & -0.9772 & -0.0271 \\
\hline
4 & 0 & -2.0009 & -0.0128 \\
 & 1--4 & -1.1048 & -0.0114 \\
 & 5--8 & -0.6181 & -0.0141 \\
\hline\hline
\end{tabular}

    \begin{minipage}{0.95\textwidth}
        \vspace{0.25cm}
        \footnotesize Note: table summarises results from \mcnsim\ replications of each data generating process. SE-TWFE is reported against the average total trajectory effect realised in each replication, calculated as the mean combined effect of all events experienced by the units observed at first-event horizon $h$ in that replication. Each quantity is formed at every horizon within a replication and then averaged, first across the horizons of the group with equal weight and then across replications.
    \end{minipage}
\end{table}
	
To assess the ability of the P-TWFE specification to recover cross-occurrence average event effects, Table~\ref{tab:mc_pooled} reports bias measures in the P-TWFE estimates compared to an average of the population first- and second-event effects $\bar\tau^{(1)}_h$ and $\bar\tau^{(2)}_h$, weighted by the number of observations contributing at each horizon in the replication:
$$\bar\tau^{(P)}_h = \frac{N^{(1)}_h \bar\tau^{(1)}_h + N^{(2)}_h \bar\tau^{(2)}_h}{N^{(1)}_h + N^{(2)}_h}.$$
This reveals that P-TWFE returns biased estimates of the cross-event average ATTs as well as of the event-specific ATTs. 

\begin{table}[H]
    \centering
    \caption{Bias of the P-TWFE Specification Against the Cross-Occurrence Average Effect}\label{tab:mc_pooled}
%
\begin{tabular}{clrr}
\hline\hline
\rule{0pt}{2.6ex}DGP & $h$ & Pooled ATT & P-TWFE bias \\
 & & $\tau^{(P)}$ & against $\tau^{(P)}$ \\
\hline
1 & 0 & -1.6668 & +0.0227 \\
 & 1--4 & -0.9332 & +0.0710 \\
 & 5--8 & -0.5532 & +0.1409 \\
\hline
2 & 0 & -1.6668 & -0.0853 \\
 & 1--4 & -0.9334 & -0.0419 \\
 & 5--8 & -0.5534 & +0.0323 \\
\hline
3 & 0 & -1.6665 & -0.0126 \\
 & 1--4 & -1.0403 & +0.0450 \\
 & 5--8 & -0.7268 & +0.1271 \\
\hline
4 & 0 & -1.6679 & +0.1868 \\
 & 1--4 & -0.8291 & +0.2629 \\
 & 5--8 & -0.4048 & +0.3542 \\
\hline\hline
\end{tabular}

    \begin{minipage}{0.95\textwidth}
        \vspace{0.25cm}
    \footnotesize Note: table summarises results from \mcnsim\ replications. $\bar\tau^{(P)}_h$ is the average of the population first- and second-event effects, weighted by the number of observations contributing at each horizon in the replication. Both $\bar\tau^{(P)}_h$ and the P-TWFE bias are formed at every horizon within a replication and then averaged, first across the horizons of the group with equal weight and then across replications.
    \end{minipage}
\end{table}

Two features of the specification are responsible. The first is that the FWL mechanics that underpin the TWFE model cause the P-TWFE specification to weight first- and second-event effects at a given horizon by the sample configuration of event timings and, in general, these weights will not coincide with the weighting that defines $\bar\tau^{(P)}_h$. This is effectively a `within-horizon' channel that causes the P-TWFE weights across occurrences to deviate from those in the pooled ATT at the horizon the pooled ATT pertains to. The second feature concerns cross-horizon contamination. Since a unit that has experienced both events is at two relative times at once, the coefficient at each horizon is fitted on outcomes that are additionally influenced by treatment effects at other horizons. Where effects differ across occurrences, a single coefficient per horizon cannot represent the effect of a specific event-horizon combination, leaving a residual that is absorbed into the coefficients at other horizons.\footnote{Proposition~\ref{prop:pooled} formalises these bias channels, showing that the second is generically non-zero and innocuous only where effects are common across occurrences. The first vanishes in the limit where conditional mean effects are common across occurrences and invariant to $(\mathbf{d},\mathbf{x})$, and otherwise arises because the weights on each occurrence at horizon $h$ need not coincide with the observation shares defining $\bar\tau^{(P)}_h$.}

Table~\ref{tab:mc_pooled} shows that the bias in the P-TWFE estimates relative to the pooled ATT grows with the horizon. Separate calculation of the two bias components described above reveals the increasing pattern is almost entirely driven by cross-horizon contamination. At short horizons after a first event, few units have yet experienced a second, so the outcomes identifying short-horizon effects can largely be attributed to a single occurrence. As the horizon grows, an increasing share of units experience a second event and the P-TWFE model must separate the contribution of each event from outcomes that reflect both. When occurrence-specific effects differ systematically, the P-TWFE model fails to reflect the true average effect of any event-horizon combination and thus imperfectly attributes effects to the horizons of each event. This attribution applies to more observations at longer horizons, causing greater cross-horizon contamination. The within-horizon channel, by contrast, is small and contributes proportionately less to the bias as the horizon grows.
	
These results have important implications for applied researchers. The pooled specification might appear a safe choice for a researcher content with an average effect across occurrences, since it asks less of the data than an event-specific estimator does. This is misleading and the model will not in general yield estimates with clear interpretation as the average effect across distinct occurrences. A researcher who wants an average across occurrences is better served by estimating the event-specific effects and averaging them, which the SIE permits at no further cost in assumptions.
	
\section{Implementation Considerations}\label{sect:implementation}
	
The Monte Carlo simulations of Section \ref{sect:montecarlo} reveal the SIE as an attractive tool for researchers interested in estimating event-specific treatment effects when treatment events can reoccur. The empirical contexts faced by such researchers will inevitably be less stylized than the controlled environment of the simulations, necessitating various implementation decisions. This section discusses two pertinent considerations.
	
\paragraph{Initial Conditions}
In many applications, datasets will not observe the full treatment trajectory $E_i$ of a unit from the time of its first existence. Units will usually predate the observed data and hence their initial conditions will be unobserved to the researcher.\footnote{Notable exceptions to the unobserved initial conditions problem are panel datasets recording individuals' full health or employment histories. Researchers who intend to examine the impact of recurrent events using such datasets can ignore the considerations discussed under the initial condition heading.} This raises two issues. 
	
First, and more substantively, it raises ambiguity over the definition of the pre-event sample partition $\Omega^{(0)}$. The SIE relies on the researcher's ability to clearly identify observations whose outcomes are unaffected by any treatment events, as this pre-treatment subsample is used to estimate unit and time fixed effects. When initial treatment trajectories are unobserved, there is the risk that observations before the first \emph{observed} event are contaminated by the effects of events occurring before the data period, which will bias unit and time fixed effect estimates and hence the SIE. 

This issue is not unique to the SIE and affects all TWFE estimators by undermining correct measurement of event-horizon indicators. Absent comprehensive data, the most straightforward way to address the issue is to assume that event effects have decayed to zero after some horizon $\bar{h}$, and to trim each unit's panel so that the earliest $\bar{h}$ periods in the estimation sample are free from treatment occurrence.\footnote{Choosing $\bar{h}=5$, for example, would correspond to assuming $\tau^{(k)}_{ih}=0$ for all $h\geq5$, $i$ and $k$.} This is similar to the solution proposed by \citet{de_chaisemartin_did_2024} to an equivalent initial conditions problem in their context of time-varying treatment, although it is less informationally costly since they also have to exclude the entire panel of units whose treatment trajectory varies at some point over the first $\bar{h}$ periods. While the choice of $\bar{h}$ is inevitably somewhat ad-hoc, it is desirable for $\bar{h}$ to be internally consistent with the treatment effect estimates that follow. If setting $\bar{h}=5$, for example, returned $\widehat{\tau}^{(k)}_{ih}\neq0$ for some $h\geq5$, it would be advisable to increase $\bar{h}$. Such iterative updating of $\bar{h}$, however, would raise issues of pre-testing and hence a prudent approach would be to set a relatively high $\bar{h}$, reporting estimates obtained under alternative $\bar{h}$ as supplementary robustness checks.
	
The second issue is semantic. If the full treatment trajectory is not observed in the data, it is unclear whether the $X$-th observed treatment event is the $X$-th treatment event ever experienced. While the count of ever-experienced events is perhaps a subject of greater research interest than the count of observed events, there is no way to overcome the ambiguity between the two concepts without fully observing units' treatment trajectories. If data lack such comprehensive coverage, it is necessary to interpret the event count indicator $k$ as the `number of treatment events experienced during the observed time period' and estimated treatment effects should accordingly be interpreted as effects of the $k$-th observed event, which need not correspond to the $k$-th event ever experienced.
	
\paragraph{Long Horizons and the Number of Modelled Events}

Assumption~\ref{ass:supp} requires adequate observations of units having experienced $m<k$ events at every prior-event horizon needed to construct counterfactuals for units hit by $k$ events. This condition becomes progressively stricter at longer horizons when event occurrence is serially correlated, and even where it holds, estimates become less precise as the samples on which the growth relationships are estimated reduce in size. Two methods to address this issue may appeal to practitioners, yet impose non-trivial restrictions.

The first approach mirrors the `binning' of event-time indicators routinely applied in TWFE models, whereby effects beyond a given horizon $\tilde{h}$ are held constant. In the SIE this amounts to setting $\lambda^{(m)}_\ell=\lambda^{(m)}_{\tilde{h}}$ for all $\ell>\tilde{h}$, so that each unit's imputed path is flat beyond $\tilde{h}$. Just as in single-event contexts, such binning is an identifying restriction 
rather than a normalisation \citep{schmidheiny_event_2023}, and where effects 
continue to evolve beyond $\tilde{h}$ it will bias estimates. The sequential structure of the SIE makes the bias introduced by terminal-horizon binning particularly consequential, since an error in the imputed event-$m$ path at long horizons is subtracted at Stage $(m+1)$.a and so will affect the event-$(m+1)$ estimate at the same observation. 

The second approach is to model fewer events than some units are observed to experience, setting $K$ below the maximum observed event count. Observations at which more than $K$ events have occurred then carry effects the specification does not parametrise, and those effects will be absorbed into the estimates of the $K$-th event. This need not undermine the utility of the SIE's estimates if, for example, analytical focus were on the effects of early event occurrences, but is an important consequence to be aware of.  

To avoid these issues, researchers should estimate a separate effect at every observed horizon and for every observed event occurrence, reporting as inestimable any average effect for which support fails.

\section{Conclusion}\label{sect:conclusion}
	
Event study research designs are widespread in applied economic research. While their performance in the context of a single absorbing treatment is now well-studied, the extent to which similar methods can be used to estimate treatment effects of recurrent events is poorly understood. This paper addresses this gap by showing a sufficiently flexible TWFE specification can obtain consistent estimates of the combined effect of a treatment trajectory under assumptions commonly invoked in single-event settings. Separate identification of the effects of each occurrence requires further assumptions. To achieve such identification, I propose a `conditional parallel effect-trends' (CPET) assumption, which mirrors the familiar parallel trend assumption of single-event difference-in-difference designs, and a sequential imputation estimator (SIE). I show the estimator is consistent when growth in event effects is linear in observed characteristics and conditionally mean independent of the subsequent event trajectory (the LEG assumption), a condition that implies CPET. Several TWFE specifications that have been applied in analyses of recurrent events, by contrast, fail to recover treatment effect estimates that can be interpreted as unbiased measures of occurrence-specific event effects under a range of plausible notions of treatment effect heterogeneity.
	
Monte Carlo simulations show the SIE performs well when its assumptions hold. A proposed diagnostic test of a selection assumption on which the SIE depends detects failures of the assumption reliably, but the test has little power against misspecification of a more specific functional form assumption. This contrasts with the single-event and pooled TWFE specifications, the most commonly adopted models in contexts of recurrent events, which exhibit large bias in estimates of event-specific treatment effects. A multi-event TWFE specification is more robust but remains vulnerable to non-negligible bias where selection into subsequent events depends on the effects of prior ones. The single-event estimator is only modestly biased for the total trajectory effect, whereas the pooled estimator exhibits substantial bias against the cross-event average effect it is designed to recover.
		
While the theoretical and simulation-based findings of this paper indicate the SIE has potential as a powerful empirical tool, its ability to answer questions of substantive empirical importance will be greater proof of its utility. Future work will address this by applying the SIE to two distinct questions of policy interest. The first application will use data on repeated minimum wage changes in the US to examine whether the effect of increases depends on the frequency between them. The second will use agricultural and climatic data from the US to quantify how the impact of drought on agricultural outcomes varies according to the intensity and frequency of droughts and examine whether agricultural insurance appears to be undermining adaptation to droughts.
	
\printbibliography
	
\appendix 
    
\section{Proof of Proposition~\ref{prop:sie}}\label{app:proof_sie}

\begin{proof}
	All probability limits are taken as $N\to\infty$ with $T$ fixed. For an event $m$ and a unit with $E^{(m)}_i<\infty$, write $\tau^{(m)}_{i\ell}$ for $\tau^{(m)}_{i,E^{(m)}_i+\ell}$ and decompose the event-$m$ effect path, following \ref{ass:leg}, as
	\[
	\tau^{(m)}_{i\ell}=\tau^{(m)}_{i0}+\mathbf{x}^{(m)\prime}_i\lambda^{(m)}_\ell+v^{(m)}_{i\ell},
	\qquad
	\E\bigl[v^{(m)}_{i\ell}\mid\mathbf{d},\mathbf{x}\bigr]=0 .
	\tag{A.1}\label{eq:leg-resid}
	\]
	For each $m$ and each observation at which the estimator constructs $\widehat\tau^{(m)}_{it}$, write the unit-event-horizon-specific error as $e^{(m)}_{it}\equiv\widehat\tau^{(m)}_{it}-\tau^{(m)}_{it}$, and collect the estimation errors that are common to all units in
	\[
	\delta_N=\bigl(\delta^{\theta\prime}_N,\delta^{\lambda\prime}_N\bigr)',
	\qquad
	\delta^{\theta}_N\equiv\widehat\theta-\theta,
	\qquad
	\delta^{\lambda}_N\equiv\bigl(\widehat\lambda^{(j)}_\ell-\lambda^{(j)}_\ell\bigr)_{j\leq K-1,\,\ell\geq1},
	\]
	where $\theta$ collects the Stage 0 period effects and their covariate interactions, as written out in Step 1. Since $T$ and $K$ are fixed, $\delta_N$ is of fixed dimension. The proof establishes the following claim.

	\medskip
	\noindent\textbf{Claim.} For each $m\in\{1,\ldots,K\}$:
	\begin{enumerate}[label=(\roman*),leftmargin=2em,itemsep=0.2ex,topsep=0.4ex]
		\item at every observation at which the estimator constructs $\widehat\tau^{(m)}_{it}$, the error decomposes as $e^{(m)}_{it}=A^{(m)}_{it}+B^{(m)}_{it}$, where $A^{(m)}_{it}$ is idiosyncratic, satisfying $\E[A^{(m)}_{it}\mid\mathbf{d},\mathbf{x}]=0$ with finite variance and independent across units, and $B^{(m)}_{it}=\mathbf{z}^{(m)\prime}_{it}\delta_N$, where $\mathbf{z}^{(m)}_{it}$ is a vector of the dimension of $\delta_N$, is a function of $(\mathbf{d},\mathbf{x})$, and satisfies $\max_{t\leq T}\E\|\mathbf{z}^{(m)}_{it}\|^2<\infty$;
		\item for $m\leq K-1$ and each $\ell\geq1$, $\widehat\lambda^{(m)}_\ell-\lambda^{(m)}_\ell=O_p(N^{-1/2})$.
	\end{enumerate}
	\medskip

	\paragraph{Step 1: Stage 0.} By Assumptions~\ref{ass:sutva}, \ref{ass:anticipation} and \ref{ass:nbc}, observations in $\Omega^{(0)}$ are untreated and their covariates unaffected by treatment, so by \ref{ass:lin} and equations~\eqref{eq:unitfe} and \eqref{eq:dgp_total},
	\[
	Y_{it}=\alpha_i+\mathbf{g}_{it}'\theta+u_{it},\qquad it\in\Omega^{(0)},
	\]
	where $\alpha_i=c_i$ as defined in \eqref{eq:unitfe} and $\theta=(\gamma_2,\ldots,\gamma_T,\pi_2',\ldots,\pi_T')'$, with $\mathbf{g}_{it}$ the corresponding vector of period indicators and their interactions with $\mathbf{x}_i$. Since $T$ is fixed, $\theta$ is of fixed dimension.

	Let $\Omega^{(0)}_i=\{t:t<E^{(1)}_i\}$, which is non-empty by the restriction $E^{(1)}_i>1$, and let a bar denote the mean over that set. Least squares with a full set of unit indicators sets the within-unit mean residual to zero over $\Omega^{(0)}_i$, so that $\widehat\alpha_i=\bar Y_i-\bar{\mathbf{g}}_i'\widehat\theta$ and hence $\widehat\alpha_i-\alpha_i=\bar u_i-\bar{\mathbf{g}}_i'\delta^{\theta}_N$. Since $\widehat{Y}^{(0)}_{it}-Y^{(0)}_{it}=(\widehat\alpha_i-\alpha_i)+\mathbf{g}_{it}'\delta^{\theta}_N-u_{it}$, this gives
	\[
	\eta_{it}\equiv\widehat{Y}^{(0)}_{it}-Y^{(0)}_{it}
	=\bar u_i-u_{it}+\widetilde{\mathbf{g}}_{it}'\delta^{\theta}_N,
	\qquad
	\widetilde{\mathbf{g}}_{it}\equiv\mathbf{g}_{it}-\bar{\mathbf{g}}_i ,
	\tag{A.2}\label{eq:eta}
	\]
	which satisfies part (i) of the claim, the first two terms being idiosyncratic and the third of the form $\mathbf{z}'\delta_N$.

	Concentrating out the unit indicators leaves a regression on $\widetilde{\mathbf{g}}_{it}$ over $\Omega^{(0)}$ in which the number of parameters is fixed while the number of observations grows at rate $N$, so $\delta^{\theta}_N=O_p(N^{-1/2})$ under \ref{ass:reg}(i)--(iv) and the mean-zero property of $u_{it}$. While the term $\bar u_i-u_{it}$ is mean zero given $(\mathbf{d},\mathbf{x})$, independent across units and of finite variance, it does not vanish as $N$ grows, because the unit effects are estimated from finitely many observations.

	The remainder of the proof establishes the claim by induction on the event index $m$. Steps 2 to 4 fix an arbitrary $m$ and establish the claim at $m$ given that it holds at every earlier event, so that verifying it at $m=1$, where no earlier event exists, establishes it in turn at $m=2$, then at $m=3$, and so on to $m=K$.

	\paragraph{Step 2: Stage $m$.a.} Fix $m\in\{1,\ldots,K\}$ and suppose the claim holds for all $j<m$. Two error identities follow directly from Definition~\ref{def:sie} and \eqref{eq:leg-resid}. For $it\in\Omega^{(m)}$ we have $K_{it}=m$ and, by the decomposition of Section~\ref{sect:estimation_target}, $\tau^{(m)}_{it}=\tau_{it}-\sum_{j=1}^{m-1}\tau^{(j)}_{it}$. Subtracting this from the Stage $m$.a estimator, or from the Stage $K$ estimator when $m=K$, and using $\widehat\tau_{it}=\tau_{it}-\eta_{it}$ gives
	\[
	e^{(m)}_{it}=-\eta_{it}-\sum_{j=1}^{m-1}e^{(j)}_{it},
	\qquad it\in\Omega^{(m)} .
	\tag{A.3}\label{eq:ema}
	\]
	For $it\in\Omega^{(k)}$ with $k>m$ and $\ell=t-E^{(m)}_i$, Stage $m$.c sets $\widehat\tau^{(m)}_{it}=\widehat\tau^{(m)}_{i0}+\mathbf{x}^{(m)\prime}_i\widehat\lambda^{(m)}_\ell$, so \eqref{eq:leg-resid} gives
	\[
	e^{(m)}_{it}=e^{(m)}_{i0}
	+\mathbf{x}^{(m)\prime}_i\bigl(\widehat\lambda^{(m)}_\ell-\lambda^{(m)}_\ell\bigr)
	-v^{(m)}_{i\ell},
	\qquad it\in\Omega^{(k)},\ k>m .
	\tag{A.4}\label{eq:emc}
	\]

	Consider \eqref{eq:ema}. Each $e^{(j)}_{it}$ with $j<m$ is supplied by Stage $j$.c, since $it\in\Omega^{(m)}$ with $m>j$, and so satisfies part (i) by the inductive hypothesis, while $\eta_{it}$ satisfies it by \eqref{eq:eta}. Summing finitely many such errors preserves the claim's two-part error structure, the idiosyncratic components remaining mean zero given $(\mathbf{d},\mathbf{x})$, independent across units and of finite variance, and the common components remaining of the form $\mathbf{z}'\delta_N$. Part (i) therefore holds for $e^{(m)}_{it}$ at every $it\in\Omega^{(m)}$. At $m=1$ the sum is empty and $e^{(1)}_{it}=-\eta_{it}$. For the impact error, that is $e^{(m)}_{it}$ at $t=E^{(m)}_i$, exactly $m$ events have
	occurred at that date so $(i,E^{(m)}_i)\in\Omega^{(m)}$. The impact error is always therefore given by \eqref{eq:ema} rather than by imputation and hence satisfies part (i).

	\paragraph{Step 3: Stage $m$.b.} This step establishes part (ii) at $m$. Let $m\leq K-1$, fix $\ell\geq1$, and write $E\equiv E^{(m)}_i$ and $S^{(m)}_\ell=\{i:(i,E+\ell)\in\Omega^{(m)}\}$, whose share of the sample converges to a positive limit by \ref{ass:supp} and \ref{ass:reg}(i), so that $|S^{(m)}_\ell|$ grows at rate $N$. For $i\in S^{(m)}_\ell$ both $E$ and $E+\ell$ lie in $\Omega^{(m)}$, so the errors in both effects entering the Stage $m$.b regressand are given by \eqref{eq:ema} and
	\[
	e^{(m)}_{i\ell}-e^{(m)}_{i0}
	=\zeta_{i\ell}-(\widetilde{\mathbf{g}}_{i,E+\ell}-\widetilde{\mathbf{g}}_{i,E})'\delta^{\theta}_N
	-\sum_{j=1}^{m-1}\bigl(e^{(j)}_{i,E+\ell}-e^{(j)}_{i,E}\bigr),
	\qquad
	\zeta_{i\ell}\equiv u_{i,E+\ell}-u_{i,E}.
	\tag{A.5}\label{eq:emb}
	\]
	The term $\bar u_i$, and with it the entire unit-effect estimation error, cancels between the two dates.

	The prior-event terms in \eqref{eq:emb} require separate treatment, since they are themselves imputed. For each $j<m$ both come from Stage $j$.c, so differencing \eqref{eq:emc} at event $j$ between horizons $\ell_1\equiv E-E^{(j)}_i$ and $\ell_1+\ell$, at which the impact error $e^{(j)}_{i0}$ cancels,
	\[
	e^{(j)}_{i,E+\ell}-e^{(j)}_{i,E}
	=\mathbf{x}^{(j)\prime}_i\Bigl[\bigl(\widehat\lambda^{(j)}_{\ell_1+\ell}-\lambda^{(j)}_{\ell_1+\ell}\bigr)-\bigl(\widehat\lambda^{(j)}_{\ell_1}-\lambda^{(j)}_{\ell_1}\bigr)\Bigr]
	-\bigl(v^{(j)}_{i,\ell_1+\ell}-v^{(j)}_{i\ell_1}\bigr).
	\tag{A.6}\label{eq:ejdiff}
	\]
	The growth-coefficient differences are $O_p(N^{-1/2})$ by the inductive hypothesis, and the residual difference is mean zero given $(\mathbf{d},\mathbf{x})$ by \eqref{eq:leg-resid}, the horizons $\ell_1$ and $\ell_1+\ell$ being functions of $\mathbf{d}$.

   	Substituting \eqref{eq:leg-resid} for the true growth and \eqref{eq:ejdiff} for the prior-event errors in \eqref{eq:emb}, the Stage $m$.b regressand equals $\mathbf{x}^{(m)\prime}_i\lambda^{(m)}_\ell+\varepsilon^{(m)}_{i\ell}$, where
	\begin{align*}
    	\varepsilon^{(m)}_{i\ell}
    	&=\underbrace{v^{(m)}_{i\ell}+\zeta_{i\ell}
    	+\sum_{j<m}\bigl(v^{(j)}_{i,\ell_1+\ell}-v^{(j)}_{i\ell_1}\bigr)}_{\textstyle a_{i\ell}}\\[1ex]
    	&\quad\underbrace{-\bigl(\widetilde{\mathbf{g}}_{i,E+\ell}-\widetilde{\mathbf{g}}_{i,E}\bigr)'\delta^{\theta}_N
    	-\sum_{j<m}\mathbf{x}^{(j)\prime}_i
    	\Bigl[\bigl(\widehat\lambda^{(j)}_{\ell_1+\ell}-\lambda^{(j)}_{\ell_1+\ell}\bigr)
    	-\bigl(\widehat\lambda^{(j)}_{\ell_1}-\lambda^{(j)}_{\ell_1}\bigr)\Bigr]}_{\textstyle b_{i\ell}},
	\end{align*}
	with $\ell_1=E-E^{(j)}_i$ depending on $j$. Ordinary least squares then gives
	\[
	\widehat\lambda^{(m)}_\ell-\lambda^{(m)}_\ell
	=\widehat{Q}^{-1}_{m\ell}\cdot\frac{1}{|S^{(m)}_\ell|}\sum_{i\in S^{(m)}_\ell}\mathbf{x}^{(m)}_i\bigl(a_{i\ell}+b_{i\ell}\bigr),
	\qquad
	\widehat{Q}_{m\ell}=\frac{1}{|S^{(m)}_\ell|}\sum_{i\in S^{(m)}_\ell}\mathbf{x}^{(m)}_i\mathbf{x}^{(m)\prime}_i .
	\tag{A.7}\label{eq:lambdahat}
	\]
	Consider the $a_{i\ell}$ part of the numerator. Each of its components is mean zero given $(\mathbf{d},\mathbf{x})$ and, since $\mathbf{x}^{(m)}_i$ is a function of $(\mathbf{d},\mathbf{x})$, $\E[\mathbf{x}^{(m)}_ia_{i\ell}\mid\mathbf{d},\mathbf{x}]=0$. Selection into $S^{(m)}_\ell$ does not disturb this, since membership too is a
	function of $\mathbf{d}$. The summands are independent across units and have finite variance under \ref{ass:reg}(ii)--(iii). Their average therefore has mean zero and variance of order $N^{-1}$, since $|S^{(m)}_\ell|$ grows at rate $N$, and is hence $O_p(N^{-1/2})$.
	
    The $b_{i\ell}$ component is a sum of terms, each a function of
	$(\mathbf{d},\mathbf{x})$ times either $\delta^{\theta}_N$ or a growth-coefficient error at some $j<m$. These factors are common to all units,
	so they do not average out; they are instead $O_p(N^{-1/2})$ by Step 1 and the inductive hypothesis. Factoring each out of the sum, the contribution of $b_{i\ell}$ to the numerator is bounded by $O_p(N^{-1/2})$ times a sample average of products of norms, which has finite expectation under \ref{ass:reg}(iii) and is hence $O_p(1)$. The numerator is therefore $O_p(N^{-1/2})$, and $\widehat{Q}_{m\ell}$ converges to a non-singular limit by \ref{ass:supp}, which gives part (ii) at $m$.
    
	\paragraph{Step 4: Stage $m$.c.} It remains to verify part (i) at the observations covered by \eqref{eq:emc}. There, $e^{(m)}_{i0}$ satisfies part (i) by Step 2, the second term is of the form $\mathbf{z}'\delta_N$ with the rate supplied by Step 3, and $v^{(m)}_{i\ell}$ is mean zero given $(\mathbf{d},\mathbf{x})$, independent across units and of finite variance by \ref{ass:reg}(iii). Part (i) therefore holds at every observation at which $\widehat\tau^{(m)}_{it}$ is constructed. At $m=K$ there is no Stage $K$.c, since no unit lies in $\Omega^{(k)}$ for $k>K$, and part (i) follows from Step 2 alone, while part (ii) is not required.

	Steps 2 to 4 therefore establish the claim at $m$ whenever it holds at every earlier event. Since it holds at $m=1$, it holds at every $m$ up to $K$. Part (ii) at every $m\leq K-1$, together with $\delta^{\theta}_N=O_p(N^{-1/2})$ from Step 1, then gives $\delta_N=O_p(N^{-1/2})$, the norm of $\delta_N$ being bounded by the sum of the norms of its finitely many blocks.

	\paragraph{Step 5: Aggregation.} Fix $m$ and weights $w_{it}$ supported on $\Omega^{(\geq m)}$ and satisfying \ref{ass:reg}(v), and write $W_i=\sum_t w_{it}$, so that $\sum_i W_i=1$ and $\max_i W_i\leq T\max_{it}w_{it}\to0$. All expectations in this step are conditional on $(\mathbf{d},\mathbf{x})$, under which the weights are non-stochastic. Writing $\bar\tau^{(m)}_{it}\equiv\E[\tau^{(m)}_{it}\mid\mathbf{d},\mathbf{x}]$, we can decompose the SIE estimation error as
	\[
	\sum_{it}w_{it}\widehat\tau^{(m)}_{it}-\sum_{it}w_{it}\bar\tau^{(m)}_{it}
	=\sum_{i}\Bigl(\sum_t w_{it}\bigl\{\tau^{(m)}_{it}-\bar\tau^{(m)}_{it}\bigr\}\Bigr)
	+\sum_{i}\Bigl(\sum_t w_{it}A^{(m)}_{it}\Bigr)
	+\sum_{it}w_{it}B^{(m)}_{it}.
	\]
	The first two terms on the RHS are sums over units of mean-zero, independent quantities, with variance bounded by $C\sum_i W_i^2\leq C\max_i W_i\to0$, the second inequality using $\sum_i W_i=1$ and $C$ being a constant that exists by \ref{ass:reg}(ii)--(iii) and the finiteness of $T$. Having mean zero and vanishing variance, both converge in probability to zero.

	The third term features $\delta_N$, which is a single draw common to all units and so does not average out across them. It vanishes instead because the common factor shrinks while the weights cannot amplify it,
	\[
	\Bigl|\sum_{it}w_{it}B^{(m)}_{it}\Bigr|
	\leq\|\delta_N\|\sum_{it}w_{it}\bigl\|\mathbf{z}^{(m)}_{it}\bigr\|
	=O_p(N^{-1/2})\cdot O_p(1)=o_p(1),
	\]
	the second factor being a weighted average of norms whose expectation is at most $\max_{t\leq T}\E\|\mathbf{z}^{(m)}_{it}\|$, finite by the claim, and hence $O_p(1)$. 
    
    All three terms therefore converge in probability to zero, so the weighted average of estimated event-$m$ effects and the corresponding weighted average of expected event-$m$ effects have the same probability limit. Part (ii) of the claim, established in Step 3, further gives $\plim\widehat\lambda^{(m)}_\ell=\lambda^{(m)}_\ell$ for each $m\leq K-1$ and each $\ell\geq1$, completing the proof.
\end{proof}
	
\section{Weighting Representations for TWFE Specifications}\label{app:twfe_proofs}
	
The three results below share a common structure and a common set of conditions. Throughout, the specification in question is estimated by OLS on a balanced panel of data generated by equation \eqref{eq:dgp_marginal} under \ref{ass:sutva}, \ref{ass:anticipation}, \ref{ass:nbc} and \ref{ass:lin}, with \ref{ass:reg} in force. Cells are indexed by $(e,\mathbf{x})$ and defined by $\mathcal{C}(e,\mathbf{x})=\{i:E_i=e,\ \mathbf{x}_i=\mathbf{x}\}$, with $N_{e,\mathbf{x}}$ denoting the number of units in the cell. In addition to \ref{ass:reg}, the results below impose that each cell retains a non-vanishing share of the sample as $N\to\infty$. This ensures that the weight attaching to any individual observation becomes uniformly small, so that the idiosyncratic components of the treatment effects average out. For Proposition~\ref{prop:multi_dynamic} I additionally assume that the within-unit-demeaned regressors of \eqref{eq:metwfe} have a positive-definite population second-moment matrix. 
	
Each result is stated relative to the following condition, which requires the conditional mean effect of each occurrence at each horizon to be the same whatever a unit's timing trajectory and covariates:
\begin{equation}\label{eq:tinv}
    \E\bigl[\tau^{(k)}_{ih}\mid\mathbf{d},\mathbf{x}\bigr]=\bar{\tau}^{(k)}_{h}
    \qquad\text{for every }k\text{ and every }h.
\end{equation}
Condition \eqref{eq:tinv} is similar to \ref{ass:cpt} but applies to the event-specific effects rather than to untreated outcomes. It is considerably stronger than \ref{ass:cpet}, which restricts conditional mean effect \emph{growth} across the $\Omega^{(m)}$ and $\Omega^{(k)}$ subsamples at a given horizon, while permitting that growth to depend on $\mathbf{x}^{(m)}_i$. 
	
\subsection{The Single-Event Dynamic Specification}\label{app:single_dynamic}
	
\begin{proposition}[Weighting representation for the SE-TWFE specification]\label{prop:single_dynamic}
    Fix a horizon $h$, write $r^{(h)}_{it}$ for the residual from projecting $w^{(1,h)}_{it}$ onto the remaining regressors of \eqref{eq:setwfe}, and define
    \[
    v^{(h)}_{it}=r^{(h)}_{it}\Big/\textstyle\sum_{js\in\Omega}\bigl(r^{(h)}_{js}\bigr)^{2}.
    \]
    Then:
    \begin{enumerate}
        \item \textbf{Weighting representation.}
        \begin{equation}\label{eq:se_weights}
            \widehat{\tau}^{SE-TWFE}_h=\sum_{h'=0}^{\bar{h}^{(1)}}\ \sum_{i\in\mathcal{N}^{(1)}_{h'}}v^{(h)}_{i,E^{(1)}_i+h'}\,\tau_{i,E^{(1)}_i+h'}+o_p(1),
        \end{equation}
        where
        \begin{equation}\label{eq:se_weightsums}
            \sum_{i\in\mathcal{N}^{(1)}_{h}}v^{(h)}_{i,E^{(1)}_i+h}=1,\qquad
            \sum_{i\in\mathcal{N}^{(1)}_{h'}}v^{(h)}_{i,E^{(1)}_i+h'}=0\ \ \text{for }h'\neq h .
        \end{equation}
        Some weights may be negative.
        
        \item \textbf{Measurability of the weights.} The weights are deterministic functions of $\Omega$, of $\mathbf{x}$ and of the sample configuration of first-event dates, and do not depend on realised outcomes. They do not depend on $E^{(k)}_i$ for any $k>1$, and $r^{(h)}_{it}=r^{(h)}_{jt}$ for every $t$ whenever $E^{(1)}_i=E^{(1)}_j$ and $\mathbf{x}_i=\mathbf{x}_j$.
        
        \item \textbf{Consistency.} Suppose that for each $h'$ there exists $\bar{\tau}_{h'}$ with
        \begin{equation}\label{eq:se_condition}
            \E\bigl[\tau_{it}\mid E^{(1)}_i,\mathbf{x}_i\bigr]=\bar{\tau}_{h'}\quad\text{for all }it\text{ with }\ell^{(1)}_{it}=h'.
        \end{equation}
        Then $\plim_{N\to\infty}\widehat{\tau}_h=\bar{\tau}_h$. Condition \eqref{eq:se_condition} holds if inter-event gaps are common across units and either event-specific effects are homogeneous across units or \eqref{eq:tinv} holds. Absent \eqref{eq:se_condition}, trajectory effects at horizons $h'\neq h$ enter $\widehat{\tau}_h$ with net weights determined by the joint distribution of $(E^{(1)}_i,\mathbf{x}_i)$, and $\widehat{\tau}_h$ need not admit interpretation as a reasonably-weighted average of total trajectory effects.
    \end{enumerate}
\end{proposition}

\begin{proof}[Proof of Proposition~\ref{prop:single_dynamic}]\leavevmode
    
\noindent\textbf{Part 1.} By \eqref{eq:dgp_total}, every term of $Y_{it}$ other than $\tau_{it}w^{(1)}_{it}$ and $u_{it}$ lies in the span of the controls in \eqref{eq:setwfe}. Applying Frisch--Waugh--Lovell to $w^{(1,h)}_{it}$ and using the orthogonality of $r^{(h)}_{it}$ to those controls,
\[
\widehat{\tau}^{SE-TWFE}_h=\sum_{it\in\Omega}v^{(h)}_{it}\,\tau_{it}w^{(1)}_{it}+\sum_{it\in\Omega}v^{(h)}_{it}u_{it}.
\]
Since $r^{(h)}_{it}$ is a deterministic function of $(\mathbf{d},\mathbf{x})$ and $\Omega$, and $\E[u_{it}\mid\mathbf{d},\mathbf{x}]=0$, the second term on the RHS has mean zero and vanishes in probability under \ref{ass:reg}. Observations with $w^{(1)}_{it}=0$ drop out of the first term, so the sum runs over $\Omega^{(\geq1)}$. Writing $w^{(1)}_{it}=\sum_{h'}w^{(1,h')}_{it}$, each product $v^{(h)}_{it}w^{(1,h')}_{it}$ is nonzero only at $t=E^{(1)}_i+h'$, so for fixed $h'$ the sum over $\Omega$ reduces to a sum over the units observed at that first-event horizon, that is, over $\mathcal{N}^{(1)}_{h'}$. Summing over $h'$ gives \eqref{eq:se_weights}.
    
For \eqref{eq:se_weightsums}, note that $w^{(1,h)}_{it}$ decomposes into $r^{(h)}_{it}$ and a fitted component orthogonal to it, so that $\sum_{it\in\Omega}r^{(h)}_{it}w^{(1,h)}_{it}=\sum_{it\in\Omega}(r^{(h)}_{it})^{2}$, and the definition of $v^{(h)}_{it}$ gives the first equality. For $h'\neq h$ the indicator $w^{(1,h')}_{it}$ is among the regressors projected out, so $\sum_{it\in\Omega}r^{(h)}_{it}w^{(1,h')}_{it}=0$, giving the second.
		
It follows that some weights must be negative. By \eqref{eq:se_weightsums} the weights $v^{(h')}_{it}$ at each horizon $h'\neq h$ sum to zero, so unless they equal zero at each $h'\neq h$, both signs must be present among them. Since each observation in $\Omega^{(\geq1)}$ belongs to exactly one horizon, the weights on $\Omega^{(\geq1)}$ therefore sum to unity while taking negative values at horizons other than $h$.
		
\medskip
\noindent\textbf{Part 2.} No indicator for an event after the first enters \eqref{eq:setwfe}, so the design matrix is determined by $\Omega$, by $\mathbf{x}$ and by the sample configuration of first-event dates. Since the projection is taken over the full sample, $r^{(h)}_{it}$ inherits this dependence, and in particular does not depend on $E^{(k)}_i$ for any $k>1$ nor on realised outcomes.
		
Regarding invariance of the residual within cells, the unit indicators can be partialled out by within-unit demeaning, so that $r^{(h)}_{it}$ is the residual from projecting the demeaned regressand on the demeaned remaining regressors. In a balanced panel every unit is observed in all $T$ periods, so the demeaned period indicators are $fs_t-T^{-1}$, their covariate interactions $(fs_t-T^{-1})\mathbf{x}_i'$, and the demeaned horizon indicator at $h'\neq h$ is $\mathbf{1}[t-E^{(1)}_i=h']-T^{-1}\mathbf{1}[E^{(1)}_i+h'\leq T]$. Each demeaned regressor, and the demeaned regressand, is therefore a function of $t$, $E^{(1)}_i$ and $\mathbf{x}_i$ alone. The projection coefficients are common across units, so the residual is a common function of these three arguments, and units agreeing in $(E^{(1)}_i,\mathbf{x}_i)$ receive identical residuals period by period.
		
\medskip
\noindent\textbf{Part 3.} By Part 2 the weight is common across units sharing $(E^{(1)}_i,\mathbf{x}_i)$ at every period. Index such groups by $(e_1,\mathbf{x})$, with $N_{e_1,\mathbf{x}}$ units, each of which is a union of cells $\mathcal{C}(e,\mathbf{x})$ and so retains a non-vanishing share of the sample. Denote the group-specific weight $v^{(h)}_{t,e_1,\mathbf{x}}$ and write $\mu_{h',e_1,\mathbf{x}}=\E[\tau_{it}\mid E^{(1)}_i=e_1,\mathbf{x}_i=\mathbf{x}]$ for the conditional mean trajectory effect at first-event horizon $h'$ within the group. All units in the group reach that horizon at the common period $t=e_1+h'$, and the group contributes $N_{e_1,\mathbf{x}}$ such observations whenever that period lies within the panel. Decomposing $\tau_{it}=\mu_{h',e_1,\mathbf{x}}+\eta_{ih'}$ and grouping the terms of \eqref{eq:se_weights} accordingly,
\[
\widehat{\tau}^{SE-TWFE}_h=\sum_{h'}\sum_{e_1,\mathbf{x}}N_{e_1,\mathbf{x}}\,v^{(h)}_{e_1+h',\,e_1,\mathbf{x}}\,\mu_{h',e_1,\mathbf{x}}
\;+\;\sum_{h'}\sum_{i\in\mathcal{N}^{(1)}_{h'}}v^{(h)}_{i,E^{(1)}_i+h'}\,\eta_{ih'}+o_p(1),
\]
the sums over groups running only over those for which $e_1+h'\leq T$. The deviations $\eta_{ih'}$ are mean zero conditional on $(E^{(1)}_i,\mathbf{x}_i)$ by construction and, by Part 2, the weights are functions of $\Omega$ and of the sample configuration of first-event dates and covariates alone. The second term on the RHS therefore has mean zero and vanishes in probability under \ref{ass:reg}.
		
Under \eqref{eq:se_condition} the conditional means depend on $h'$ alone and may be taken outside the sum over groups, leaving
\[
\widehat{\tau}^{SE-TWFE}_h=\sum_{h'}\bar{\tau}_{h'}\sum_{e_1,\mathbf{x}}N_{e_1,\mathbf{x}}\,v^{(h)}_{e_1+h',\,e_1,\mathbf{x}}+o_p(1).
\]
The remaining sum adds the weight of every observation at first-event horizon $h'$, grouped by $(e_1,\mathbf{x})$, and so equals the corresponding sum in \eqref{eq:se_weightsums}. It is therefore unity at $h'=h$ and zero otherwise, which yields $\plim_{N\to\infty}\widehat{\tau}^{SE-TWFE}_h=\bar{\tau}_h$.

For the sufficient conditions, write $G^{(k)}_i=E^{(k)}_i-E^{(1)}_i$ for the gap between unit $i$'s first and $k$-th events, with $G^{(1)}_i=0$ and $G^{(k)}_i=\infty$ where fewer than $k$ events occur, so that at first-event horizon $h'$
\begin{equation}\label{eq:se_active}
    \tau_{it}=\sum_{k=1}^{K}\tau^{(k)}_{i,h'-G^{(k)}_i}\mathbf{1}\bigl[h'\geq G^{(k)}_i\bigr].
\end{equation}
Where gaps are common, $G^{(k)}_i=G^{(k)}$, the same set of events is active at $h'$ for every unit and the indicators in \eqref{eq:se_active} do not vary across units. If in addition each $\tau^{(k)}_{i\ell}=\tau^{(k)}_\ell$, then \eqref{eq:se_active} is a deterministic function of $h'$. If instead \eqref{eq:tinv} holds, then
    \[
    \E\bigl[\tau_{it}\mid\mathbf{d},\mathbf{x}\bigr]=\sum_{k=1}^{K}\bar{\tau}^{(k)}_{h'-G^{(k)}}\mathbf{1}\bigl[h'\geq G^{(k)}\bigr],
    \]
    which again depends on $h'$ alone. In either case $\E[\tau_{it}\mid\mathbf{d},\mathbf{x}]$ depends on $h'$ alone, so by iterated expectations so does $\E[\tau_{it}\mid E^{(1)}_i,\mathbf{x}_i]$, and \eqref{eq:se_condition} follows.
    
    Absent \eqref{eq:se_condition}, the conditional means differ across the groups contributing at a given horizon $h'\neq h$. The weight sums \eqref{eq:se_weightsums} hold across groups but not group by group, so no common mean can be factored, and the horizon-$h'$ contribution is generically nonzero with net weights determined by the joint distribution of $(E^{(1)}_i,\mathbf{x}_i)$.
    
\end{proof}

\begin{remark}[Role of the gap distribution]
Let $\tau^{(k)}_{i\ell}=\tau^{(k)}_\ell$ for every $i$ and $k$. Then $$\E[\tau_{it}\mid E^{(1)}_i,\mathbf{x}_i]=\sum_k\E\bigl[\tau^{(k)}_{h'-G^{(k)}_i}\mathbf{1}[h'\geq G^{(k)}_i]\mid E^{(1)}_i,\mathbf{x}_i\bigr],$$ which is invariant to $(E^{(1)}_i,\mathbf{x}_i)$ if the distribution of inter-event gaps is, and will in general vary with them otherwise. Homogeneous effects are therefore not sufficient when the distribution of inter-event gaps varies with the first-event date, as it does whenever right-truncation of the panel restricts the gaps available to late first-event cohorts.
\end{remark}
	
\subsection{The Pooled Dynamic Specification}\label{app:pooled_dynamic}
	
\begin{proposition}[Weighting representation for the P-TWFE specification]\label{prop:pooled}
	Fix a horizon $h$, write $r^{(h)}_{it}$ for the residual from projecting $\bar{w}^{(h)}_{it}$ onto the remaining regressors of \eqref{eq:ptwfe}, and define
    \[
    v^{(h)}_{it}=r^{(h)}_{it}\Big/\textstyle\sum_{js\in\Omega}\bigl(r^{(h)}_{js}\bigr)^{2},
    \qquad
    A^{(k)}_{h'}=\sum_{it\in\Omega}v^{(h)}_{it}\,w^{(k,h')}_{it},
    \]
    the latter denoting the total weight that $\widehat{\tau}_h$ places on occurrence $k$ at horizon $h'$. Then:
    \begin{enumerate}
        \item \textbf{Weighting representation.}
        \begin{equation}\label{eq:pooled_weights}
            \widehat{\tau}^{P-TWFE}_h=\sum_{h'=0}^{\bar{h}^{(1)}}\ \sum_{k=1}^{K}\ \sum_{i\in\mathcal{N}^{(k)}_{h'}}v^{(h)}_{i,E^{(k)}_i+h'}\,\tau^{(k)}_{ih'}+o_p(1),
        \end{equation}
        where the occurrence-specific weight totals satisfy
        \begin{equation}\label{eq:pooled_weightsums}
            \sum_{k=1}^{K}A^{(k)}_{h}=1,\qquad \sum_{k=1}^{K}A^{(k)}_{h'}=0\ \ \text{for }h'\neq h .
        \end{equation}
        Some weights on treated observations may be negative, and the individual $A^{(k)}_{h'}$ are unrestricted beyond \eqref{eq:pooled_weightsums}.
        
        \item \textbf{Measurability of the weights.} The weights are deterministic functions of $\Omega$, of $\mathbf{x}$ and of the sample configuration of event timings, and do not depend on realised outcomes. Since every occurrence enters \eqref{eq:ptwfe}, they depend on the complete trajectory $E_i$ rather than on $E^{(1)}_i$ alone, and $r^{(h)}_{it}=r^{(h)}_{jt}$ for every $t$ whenever $E_i=E_j$ and $\mathbf{x}_i=\mathbf{x}_j$.
        
        \item \textbf{Consistency.} Suppose \eqref{eq:tinv} holds with $\bar{\tau}^{(k)}_{h'}=\bar{\tau}_{h'}$ for every $k$, so that conditional mean effects are invariant to the timing trajectory and to the covariates and are common across occurrences. Then $\plim_{N\to\infty}\widehat{\tau}^{P-TWFE}_h=\bar{\tau}_h$.
        
        \item \textbf{Failure under occurrence heterogeneity.} Suppose instead that \eqref{eq:tinv} holds but that $\bar{\tau}^{(k)}_{h'}$ varies across occurrences. Then
        \begin{equation}\label{eq:pooled_contam}
            \widehat{\tau}^{P-TWFE}_h=\underbrace{\sum_{k=1}^{K}\bar{\tau}^{(k)}_{h}A^{(k)}_{h}}_{\text{own horizon}}
            \;+\;\underbrace{\sum_{\substack{h'\geq0\\h'\neq h}}\ \sum_{k=2}^{K}\bigl(\bar{\tau}^{(k)}_{h'}-\bar{\tau}^{(1)}_{h'}\bigr)A^{(k)}_{h'}}_{\text{cross-horizon contamination}}+o_p(1).
        \end{equation}
        The first term combines the occurrence-specific effects at horizon $h$ with weights totalling unity whose signs are unconstrained by \eqref{eq:pooled_weightsums}. The second is generically nonzero and vanishes at every $h$ only if effects are common across occurrences at every horizon.
    \end{enumerate}
\end{proposition}
	
\begin{proof}[Proof of Proposition~\ref{prop:pooled}]\leavevmode
		
\noindent\textbf{Part 1.} By \eqref{eq:dgp_marginal}, every term of $Y_{it}$ other than the event-specific effects and $u_{it}$ lies in the span of the controls in \eqref{eq:ptwfe}. Applying Frisch--Waugh--Lovell to $\bar{w}^{(h)}_{it}$ and using the orthogonality of $r^{(h)}_{it}$ to those controls,
\[
\widehat{\tau}^{P-TWFE}_h=\sum_{it\in\Omega}v^{(h)}_{it}\sum_{k=1}^{K}\sum_{h'\geq0}\tau^{(k)}_{ih'}\,w^{(k,h')}_{it}+\sum_{it\in\Omega}v^{(h)}_{it}u_{it}.
\]
Since $r^{(h)}_{it}$ is a deterministic function of $(\mathbf{d},\mathbf{x})$ and $\Omega$, and $\E[u_{it}\mid\mathbf{d},\mathbf{x}]=0$, the second term on the RHS has mean zero and vanishes in probability under \ref{ass:reg}. Each product $v^{(h)}_{it}w^{(k,h')}_{it}$ is nonzero only at $t=E^{(k)}_i+h'$, so for fixed $(k,h')$ the sum over $\Omega$ reduces to a sum over the units observed at horizon $h'$ after their $k$-th event, that is, over $\mathcal{N}^{(k)}_{h'}$. Summing over $(k,h')$ gives \eqref{eq:pooled_weights}.

For \eqref{eq:pooled_weightsums}, summing the definition of $A^{(k)}_{h'}$ over occurrences and using $\sum_{k}w^{(k,h')}_{it}=\bar{w}^{(h')}_{it}$ gives
\[
\sum_{k=1}^{K}A^{(k)}_{h'}=\sum_{it\in\Omega}v^{(h)}_{it}\,\bar{w}^{(h')}_{it},
\]
so the occurrence-specific weight totals sum to the weight the specification places on the pooled indicator at horizon $h'$. At $h'=h$ this equals unity, since $\bar{w}^{(h)}_{it}$ decomposes into $r^{(h)}_{it}$ and a fitted component orthogonal to it, so that $\sum_{it\in\Omega}r^{(h)}_{it}\bar{w}^{(h)}_{it}=\sum_{it\in\Omega}(r^{(h)}_{it})^{2}$, and the definition of $v^{(h)}_{it}$ then delivers the result. At $h'\neq h$ it equals zero, since $\bar{w}^{(h')}_{it}$ is among the regressors projected out. The individual $A^{(k)}_{h'}$ are unrestricted, since $w^{(k,h')}_{it}$ is not itself a regressor of \eqref{eq:ptwfe}. By \eqref{eq:pooled_weightsums} the occurrence-specific weight totals at each horizon $h'\neq h$ sum to zero so, unless all weights are identically zero, some must be negative.
		
\medskip
\noindent\textbf{Part 2.} Every occurrence indicator $w^{(k,h)}_{it}$ enters \eqref{eq:ptwfe} through the pooled regressors, so the design matrix is determined by $\Omega$, by $\mathbf{x}$ and by the sample configuration of complete trajectories. Since the projection is taken over the full sample, $r^{(h)}_{it}$ inherits this dependence and does not depend on realised outcomes. 
		
Turning to invariance within cells, the unit indicators can be partialled out by within-unit demeaning, so that $r^{(h)}_{it}$ is the residual from projecting the demeaned regressand on the demeaned remaining regressors. In a balanced panel every unit is observed in all $T$ periods, so the demeaned period indicators are $fs_t-T^{-1}$, their covariate interactions $(fs_t-T^{-1})\mathbf{x}_i'$, and the demeaned pooled indicator at $h'$ is $\bar{w}^{(h')}_{it}-T^{-1}\sum_{k}\mathbf{1}[E^{(k)}_i+h'\leq T]$. Each demeaned regressor, and the demeaned regressand, is therefore a function of $t$, $E_i$ and $\mathbf{x}_i$ alone. The projection coefficients are common across units, so the residual is a common function of these three arguments, and units agreeing in $(E_i,\mathbf{x}_i)$ receive identical residuals period by period.
		
\medskip
\noindent\textbf{Part 3.} By Part 2 the weight is common within each cell at every period. Denote this cell-specific weight $v^{(h)}_{t,e,\mathbf{x}}$. Since all units in a cell share the trajectory $e$, the observation at horizon $h'$ after the $k$-th event falls at the common period $t=e^{(k)}+h'$, and the cell contributes $N_{e,\mathbf{x}}$ such observations whenever that period lies within the panel. Decomposing $\tau^{(k)}_{ih'}=\mu^{(k)}_{h',e,\mathbf{x}}+\eta^{(k)}_{ih'}$, where $\mu^{(k)}_{h',e,\mathbf{x}}=\E[\tau^{(k)}_{ih'}\mid\mathbf{d},\mathbf{x}]$, and grouping the terms of \eqref{eq:pooled_weights} by cell,
    \[
    \widehat{\tau}^{P-TWFE}_h=\sum_{h'}\sum_{k=1}^{K}\sum_{e,\mathbf{x}}N_{e,\mathbf{x}}\,v^{(h)}_{e^{(k)}+h',\,e,\mathbf{x}}\,\mu^{(k)}_{h',e,\mathbf{x}}
    \;+\;\sum_{h'}\sum_{k=1}^{K}\sum_{i\in\mathcal{N}^{(k)}_{h'}}v^{(h)}_{i,E^{(k)}_i+h'}\,\eta^{(k)}_{ih'}+o_p(1),
    \]
the sums over cells running only over those for which $e^{(k)}+h'\leq T$. The deviations $\eta^{(k)}_{ih'}$ are mean zero conditional on $(\mathbf{d},\mathbf{x})$ by construction and the weights are functions of $(\mathbf{d},\mathbf{x})$ and $\Omega$ by Part 2, so the second term has mean zero and vanishes in probability under \ref{ass:reg}.
		
Under the condition of Part 3 the conditional means equal $\bar{\tau}_{h'}$ for every $k$ and every cell, so they may be taken outside the sums over $k$ and over cells, leaving
\[
\widehat{\tau}^{P-TWFE}_h=\sum_{h'}\bar{\tau}_{h'}\sum_{k=1}^{K}\sum_{e,\mathbf{x}}N_{e,\mathbf{x}}\,v^{(h)}_{e^{(k)}+h',\,e,\mathbf{x}}+o_p(1)
=\sum_{h'}\bar{\tau}_{h'}\sum_{k=1}^{K}A^{(k)}_{h'}+o_p(1),
\]
the second equality holding because the cell sum at each $(k,h')$ adds the weight of every observation at horizon $h'$ after a $k$-th event, which is $A^{(k)}_{h'}$. By \eqref{eq:pooled_weightsums} the inner sum equals unity at $h'=h$ and zero otherwise, so $\plim_{N\to\infty}\widehat{\tau}^{P-TWFE}_h=\bar{\tau}_h$.
		
\medskip
\noindent\textbf{Part 4.} Where $\bar{\tau}^{(k)}_{h'}$ varies across occurrences, the conditional means may still be taken outside the sum over cells at each $(k,h')$, giving
\[
\widehat{\tau}^{P-TWFE}_h=\sum_{h'}\sum_{k=1}^{K}\bar{\tau}^{(k)}_{h'}A^{(k)}_{h'}+o_p(1).
\]
Separating the terms at $h'=h$ delivers the first term of \eqref{eq:pooled_contam}, whose weights total unity by \eqref{eq:pooled_weightsums} (but which does not constrain their signs). For each $h'\neq h$, substituting $A^{(1)}_{h'}=-\sum_{k\geq2}A^{(k)}_{h'}$ from \eqref{eq:pooled_weightsums} gives
    \[
    \sum_{k=1}^{K}\bar{\tau}^{(k)}_{h'}A^{(k)}_{h'}
    =\sum_{k=2}^{K}\bigl(\bar{\tau}^{(k)}_{h'}-\bar{\tau}^{(1)}_{h'}\bigr)A^{(k)}_{h'},
    \]
which is the second term of \eqref{eq:pooled_contam}. The occurrence-specific effects therefore enter only through their departures from those of the first event, with the choice of the first occurrence as the base a normalisation. The term vanishes at every $h$ under either of two conditions. The first is $\bar{\tau}^{(k)}_{h'}=\bar{\tau}^{(1)}_{h'}$ for all $k$ and $h'$, which is homogeneity across occurrences and returns the setting of Part 3. The second is $A^{(k)}_{h'}=0$ for all $k\geq2$ and $h'\neq h$, which requires $r^{(h)}_{it}$ to be orthogonal to each occurrence-specific indicator individually. The projection in \eqref{eq:ptwfe} imposes orthogonality only to the pooled indicators, and so does not deliver this.
\end{proof}
	
	
\subsection{The Multi-Event Dynamic Specification}\label{app:multi_dynamic}
	
\begin{proposition}[Weighting representation for the ME-TWFE specification]\label{prop:multi_dynamic}
Fix an occurrence $k$ and a horizon $h$, write $r^{(k,h)}_{it}$ for the residual from projecting $w^{(k,h)}_{it}$ onto the remaining regressors of \eqref{eq:metwfe}, and define
\[
v^{(k,h)}_{it}=r^{(k,h)}_{it}\Big/\textstyle\sum_{js\in\Omega}\bigl(r^{(k,h)}_{js}\bigr)^{2},
\qquad
A^{(k',h')}=\sum_{it\in\Omega}v^{(k,h)}_{it}\,w^{(k',h')}_{it}.
\]
Then:
\begin{enumerate}
    \item \textbf{Weighting representation.}
    \begin{equation}\label{eq:me_weights}
        \widehat{\tau}^{(k)ME-TWFE}_h=\sum_{k'=1}^{K}\ \sum_{h'=0}^{\bar{h}^{(k')}}\ \sum_{i\in\mathcal{N}^{(k')}_{h'}}v^{(k,h)}_{i,E^{(k')}_i+h'}\,\tau^{(k')}_{ih'}+o_p(1),
    \end{equation}
    where
    \begin{equation}\label{eq:me_weightsums}
        A^{(k,h)}=1,\qquad A^{(k',h')}=0\ \ \text{for every }(k',h')\neq(k,h).
    \end{equation}
    Some weights may be negative.
			
    \item \textbf{Measurability of the weights.} The weights are deterministic functions of $\Omega$, of $\mathbf{x}$ and of the sample configuration of event timings, and do not depend on realised outcomes. Since every occurrence enters \eqref{eq:metwfe} separately, they depend on the complete trajectory $E_i$ rather than on the timing of any single event in isolation, and $r^{(k,h)}_{it}=r^{(k,h)}_{jt}$ for every $t$ whenever $E_i=E_j$ and $\mathbf{x}_i=\mathbf{x}_j$.
    
	\item \textbf{Consistency.} Suppose \eqref{eq:tinv} holds. Then $\plim_{N\to\infty}\widehat{\tau}^{(k)ME-TWFE}_h=\bar{\tau}^{(k)}_h$.
			
	\item \textbf{Failure under trajectory-dependent heterogeneity.} Absent \eqref{eq:tinv},
    \begin{equation}\label{eq:me_contam}
        \widehat{\tau}^{(k)ME-TWFE}_h=\sum_{e,\mathbf{x}}\theta^{(k,h)}_{h,e,\mathbf{x}}\,\mu^{(k)}_{h,e,\mathbf{x}}
        \;+\;\sum_{(k',h')\neq(k,h)}\ \sum_{e,\mathbf{x}}\theta^{(k',h')}_{h',e,\mathbf{x}}\,\mu^{(k')}_{h',e,\mathbf{x}}+o_p(1),
    \end{equation}
    where $\mu^{(k')}_{h',e,\mathbf{x}}=\E[\tau^{(k')}_{ih'}\mid\mathbf{d},\mathbf{x}]$ and $\theta^{(k',h')}_{h',e,\mathbf{x}}$ denotes the total weight that $\widehat{\tau}^{(k)}_h$ places on the horizon-$h'$ observations of occurrence $k'$ among units in the cell. The second sum has weights totalling zero at each $(k',h')$ but is generically nonzero, while the first has weights totalling unity whose signs are unconstrained by \eqref{eq:me_weightsums}.
\end{enumerate}
\end{proposition}
	
\begin{proof}[Proof of Proposition~\ref{prop:multi_dynamic}]\leavevmode
    
\noindent\textbf{Part 1.} By \eqref{eq:dgp_marginal}, every term of $Y_{it}$ other than the event-specific effects and $u_{it}$ lies in the span of the controls in \eqref{eq:metwfe}. Applying Frisch--Waugh--Lovell to $w^{(k,h)}_{it}$ and using the orthogonality of $r^{(k,h)}_{it}$ to those controls,
\[
\widehat{\tau}^{(k)}_h=\sum_{it\in\Omega}v^{(k,h)}_{it}\sum_{k'=1}^{K}\sum_{h'\geq0}\tau^{(k')}_{ih'}\,w^{(k',h')}_{it}+\sum_{it\in\Omega}v^{(k,h)}_{it}u_{it}.
\]
Since $r^{(k,h)}_{it}$ is a deterministic function of $(\mathbf{d},\mathbf{x})$ and $\Omega$, and $\E[u_{it}\mid\mathbf{d},\mathbf{x}]=0$, the second term on the RHS has mean zero and vanishes in probability under \ref{ass:reg}. Each product $v^{(k,h)}_{it}w^{(k',h')}_{it}$ is nonzero only at $t=E^{(k')}_i+h'$, so for fixed $(k',h')$ the sum over $\Omega$ reduces to a sum over $\mathcal{N}^{(k')}_{h'}$. Summing over $(k',h')$ gives \eqref{eq:me_weights}.
		
For \eqref{eq:me_weightsums}, note that $w^{(k,h)}_{it}$ decomposes into $r^{(k,h)}_{it}$ and a fitted component orthogonal to it, so that $\sum_{it\in\Omega}r^{(k,h)}_{it}w^{(k,h)}_{it}=\sum_{it\in\Omega}(r^{(k,h)}_{it})^{2}$, and the definition of $v^{(k,h)}_{it}$ gives $A^{(k,h)}=1$. Every other indicator $w^{(k',h')}_{it}$ is itself a regressor of \eqref{eq:metwfe} and is therefore among those projected out, so $A^{(k',h')}=0$. Since $A^{(k',h')}=0$ at every $(k',h')\neq(k,h)$, the weights at those occurrence–horizon combinations cannot all be non-negative unless they are identically zero.
		
\medskip
\noindent\textbf{Part 2.} Every occurrence indicator enters \eqref{eq:metwfe} separately, so the design matrix is determined by $\Omega$, by $\mathbf{x}$ and by the sample configuration of complete trajectories. Since the projection is taken over the full sample, $r^{(k,h)}_{it}$ inherits this dependence and does not depend on realised outcomes.
		
Regarding invariance within cells, the unit indicators can be partialled out by within-unit demeaning, so that $r^{(k,h)}_{it}$ is the residual from projecting the demeaned regressand on the demeaned remaining regressors. In a balanced panel every unit is observed in all $T$ periods, so the demeaned period indicators are $fs_t-T^{-1}$, their covariate interactions $(fs_t-T^{-1})\mathbf{x}_i'$, and the demeaned occurrence indicator at $(k',h')$ is $\mathbf{1}[t-E^{(k')}_i=h']-T^{-1}\mathbf{1}[E^{(k')}_i+h'\leq T]$. Each demeaned regressor, and the demeaned regressand, is therefore a function of $t$, $E_i$ and $\mathbf{x}_i$ alone. The projection coefficients are common across units, so the residual is a common function of these three arguments, and units agreeing in $(E_i,\mathbf{x}_i)$ receive identical residuals period by period.

\medskip
\noindent\textbf{Part 3.} By Part 2 the weight is common within each cell at every period and we can use $v^{(k,h)}_{t,e,\mathbf{x}}$ to denote the cell-specific weight. Since all units in a cell share the trajectory $e$, the observation at horizon $h'$ after the $k'$-th event falls at the common period $t=e^{(k')}+h'$, and the cell contributes $N_{e,\mathbf{x}}$ such observations whenever that period lies within the panel. Decomposing $\tau^{(k')}_{ih'}=\mu^{(k')}_{h',e,\mathbf{x}}+\eta^{(k')}_{ih'}$ and grouping the terms of \eqref{eq:me_weights} by cell,
\[
\widehat{\tau}^{(k)ME-TWFE}_h=\sum_{k'=1}^{K}\sum_{h'}\sum_{e,\mathbf{x}}N_{e,\mathbf{x}}\,v^{(k,h)}_{e^{(k')}+h',\,e,\mathbf{x}}\,\mu^{(k')}_{h',e,\mathbf{x}}
\;+\;\sum_{k'=1}^{K}\sum_{h'}\sum_{i\in\mathcal{N}^{(k')}_{h'}}v^{(k,h)}_{i,E^{(k')}_i+h'}\,\eta^{(k')}_{ih'}+o_p(1),
\]
the sums over cells running only over those for which $e^{(k')}+h'\leq T$. The deviations $\eta^{(k')}_{ih'}$ are mean zero conditional on $(\mathbf{d},\mathbf{x})$ by construction and the weights are functions of $(\mathbf{d},\mathbf{x})$ and $\Omega$ by Part 2, so the second term on the RHS has mean zero and vanishes in probability under \ref{ass:reg}.
		
Under \eqref{eq:tinv} the conditional means are constant across cells and may therefore be taken outside the sum over $(e,\mathbf{x})$, leaving
\[
\widehat{\tau}^{(k)}_h=\sum_{k'=1}^{K}\sum_{h'}\bar{\tau}^{(k')}_{h'}\sum_{e,\mathbf{x}}N_{e,\mathbf{x}}\,v^{(k,h)}_{e^{(k')}+h',\,e,\mathbf{x}}+o_p(1).
\]
The remaining sum adds the weight of every observation at horizon $h'$ after a $k'$-th event, grouped by cell, and so equals $A^{(k',h')}$. By \eqref{eq:me_weightsums} it is unity at $(k',h')=(k,h)$ and zero otherwise, which yields $\plim_{N\to\infty}\widehat{\tau}^{(k)ME-TWFE}_h=\bar{\tau}^{(k)}_h$.

\medskip
\noindent\textbf{Part 4.} Absent \eqref{eq:tinv} the conditional means vary across cells and cannot be taken outside the sum over $(e,\mathbf{x})$. Writing
\[
\theta^{(k',h')}_{h',e,\mathbf{x}}=N_{e,\mathbf{x}}\,v^{(k,h)}_{e^{(k')}+h',\,e,\mathbf{x}}
\]
and separating the term at $(k',h')=(k,h)$, the first display of Part 3 delivers \eqref{eq:me_contam}. The weight totals satisfy $\sum_{e,\mathbf{x}}\theta^{(k',h')}_{h',e,\mathbf{x}}=A^{(k',h')}$, which equals unity at $(k,h)$ and zero elsewhere by \eqref{eq:me_weightsums}. Constancy of $\mu^{(k')}_{h',e,\mathbf{x}}$ across the cells contributing at a given $(k',h')$ is sufficient for the corresponding inner sum to vanish, since the common value then multiplies $A^{(k',h')}=0$. Absent \eqref{eq:tinv} no such factoring is available and the second sum of \eqref{eq:me_contam} is generically nonzero. The first sum places total weight unity on the cell means at $(k,h)$, but \eqref{eq:me_weightsums} constrains only that total and not the individual $\theta^{(k,h)}_{h,e,\mathbf{x}}$, so $\widehat{\tau}^{(k)}_h$ need not be a convex combination of the cell means at its own occurrence and horizon.
\end{proof}
	
\begin{remark}[Rank condition]\label{rem:me_rank}
    Estimation of \eqref{eq:metwfe} requires inter-event gaps to vary across units. Where every doubly-treated unit shares the gap $G$, the indicators $w^{(1,h)}_{it}$ and $w^{(2,h-G)}_{it}$ coincide at every observation for which both are defined and the specification is not estimable. For $K>2$ the same applies to every pair of occurrences at their common gap. Identification therefore rests on variation in the timing of subsequent events relative to earlier ones.
\end{remark}
	
	
\end{document}